%% file: main.tex
\documentclass[runningheads,final]{llncs}
\usepackage{amsmath, amssymb}
\usepackage{stmaryrd}
\usepackage{mathtools}
\usepackage{paralist}
\usepackage{graphicx}
\usepackage{xcolor}
\usepackage{alltt}

\usepackage{tikz}
\usetikzlibrary{automata, positioning, arrows}
\tikzset{
->, %
node distance=2cm, %
initial text=$ $, %
state/.style={circle, draw, minimum size=0.2cm}
}

\usepackage{cite}
\usepackage{xspace}
\usepackage{booktabs}

\usepackage[normalem]{ulem}

\usepackage{mymacros}
\makeatletter
\RequirePackage[bookmarks,unicode,colorlinks=true]{hyperref}%
   \def\@citecolor{blue}%
   \def\@urlcolor{blue}%
   \def\@linkcolor{blue}%

\def\orcidID#1{\smash{\href{http://orcid.org/#1}{\protect\raisebox{-1.25pt}{\protect\includegraphics{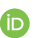}}}}}
\makeatother
\author{%
Takumi Shimoda\inst{1} \and
Naoki Kobayashi\inst{1}\orcidID{0000-0002-0537-0604} \and
Ken Sakayori\inst{1}\orcidID{0000-0003-3238-9279} \and
Ryosuke Sato\inst{1}\orcidID{0000-0001-8679-2747}%
}
\authorrunning{T. Shimoda et al.}
\institute{The University of Tokyo, Tokyo, Japan}

\begin{document}
\title{Symbolic Automatic Relations and\texorpdfstring{\\}{} Their Applications to SMT and CHC Solving}
\titlerunning{Symbolic Automatic Relations}
\maketitle              %
\input{abstract}
\input{introduction}
\input{preliminaries}
\input{sar}
\input{undecidability}
\input{reduction}
\input{experiments}
\input{related-work}

\input{conclusion}

\newpage

\bibliographystyle{splncs04}
\bibliography{main}

\newpage
\appendix
\input{reduction-appx}
\clearpage
\input{exp_detail}

\end{document}

%% file: abstract.tex
\begin{abstract}

  Despite the recent advance of automated program verification,
  reasoning about recursive data structures remains as a challenge for
  verification tools and their backends such as SMT and CHC solvers.
  To address the challenge, we introduce the notion of symbolic automatic
  relations (SARs), which combines symbolic automata and automatic relations,
  and inherits their good properties such as the closure under Boolean operations.
  We consider the satisfiability problem for SARs, and show that it is undecidable
  in general, but that we can construct a sound (but incomplete)
  and automated satisfiability checker by a reduction to CHC solving.
  We discuss applications to SMT and CHC solving on data structures, and show
  the effectiveness of our approach through experiments.
\end{abstract}

%% file: introduction.tex
\section{Introduction}
\label{sec:intro}

The recent advance of automated or semi-automated program verification
tools owes much to the improvement of SMT (Satisfiability Modulo
Theory) and CHC (Constrained Horn Clauses) solvers. The
former~\cite{MouraBjorner08,BarrettCDHJKRT11} can automatically check
the satisfiability of quantifier-free formulas modulo background
theories (such as linear integer arithmetic), and the
latter~\cite{KomuravelliGC16,ChampionCKS20,HojjatRummer18} can
automatically reason about recursively defined predicates (which can be used
to model loops and recursive functions).
Various program verification problems can be reduced to CHC
solving~\cite{Bjorner15}.
The current SMT and CHC solvers are, however,  not
very good at reasoning about recursive data structures (such as lists and trees),
compared with the capability of reasoning about basic data such as integers
and real numbers.
Indeed, improving the treatment of recursive data structures
has recently been an active research topic, especially for CHC solvers~\cite{DBLP:journals/tplp/AngelisFPP18a,UnnoTS17,ChampionCKS20,DBLP:conf/tacas/FedyukovichE21}.

In the present paper, we propose an automata-based approach for checking the satisfiability
of formulas over recursive data structures.
(For the sake of simplicity, we focus on
lists of integers; our approach can, in principle, be extended for more general
data structures).
More precisely, we introduce the notion of \emph{symbolic automatic relations}, which is obtained
by combining \emph{automatic relations}~\cite{BlumensathGradel00} and
\emph{symbolic automata}~\cite{VeanesHT10,DBLP:conf/lpar/VeanesBM10,DBLP:conf/cav/DAntoniV17}.

A \(k\)-ary automatic relation is a relation on \(k\) words\footnote{We use ``lists'',
  ``words'', and ``sequences'' interchangeably.}
that can be recognized by a finite state automaton
that reads \(k\) words in a synchronous manner
(so, given \(k\) words, \(x_{01}\cdots x_{0m}, \ldots, x_{(k-1)1}\cdots x_{(k-1)m}\),
the automaton reads a tuple \((x_{0i},\ldots,x_{(k-1)i})\) at each transition;
if the input words have different lengths,
the special padding symbol \(\pad\) is filled at the
end).
For example, the equality relation on two words over the alphabet \(\set{a,b}\) is an automatic
relation, since it is recognized by the automaton with a single state \(q\) (which
is both initial and accepting) with the transition \(\delta(q, (a,a))=q\) and
\(\delta(q,(b,b))=q\).
By using automatic relations, we can express and manipulate relations on data structures.

\begin{figure}[t]
\begin{minipage}{0.3\linewidth}
\fbox{\( M_{\text{sync}} \)}
\par\vspace{-1em}
\centering
\begin{tikzpicture}
\node[state, initial, accepting] (q) {\(q\)};
\draw (q) edge[loop above] node{\( (a, a) \)} (q);
\draw (q) edge[loop right] node{\( (b, b) \)} (q);
\end{tikzpicture}
\end{minipage}
\begin{minipage}{0.4\linewidth}
\fbox{\( M_{\text{symb}} \)}
\par\vspace{-1em}
\centering
\begin{tikzpicture}
\node[state, initial,accepting] (q0) {\( q_0 \)};
\node[state, accepting,right of=q0, xshift=1cm] (q1) {\( q_1 \)};
\draw (q0) edge[above, bend left=20] node{\( 0 < x \)} (q1)
      (q1) edge[below, bend left=20] node{\( x < 0 \)} (q0);
\end{tikzpicture}
\end{minipage}
\begin{minipage}{0.25\linewidth}
\fbox{\( M_< \)}
\par\vspace{-1em}
\centering
\begin{tikzpicture}
\node[state, initial, accepting] (q) {\(q\)};
\draw (q) edge[loop above] node{\( l_0 < l_1 \)} (q);
\end{tikzpicture}
\end{minipage}
\caption{Examples of synchronous, symbolic and symbolic synchronous automaton.}
\label{fig:automata-for-introduction}
\end{figure}
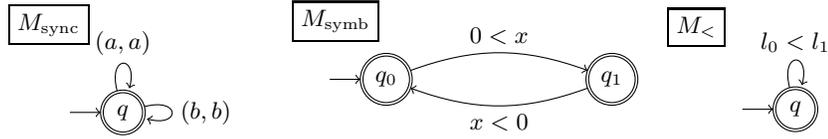

Since data structures typically contain elements from an infinite set,
we extend automatic relations by using symbolic automata.
Here, a symbolic automaton is a variation of finite state automaton whose alphabet is
possibly infinite, and whose transition is described by
a formula over elements of the alphabet.
For example, \( M_{\text{symb}} \) on Fig.~\ref{fig:automata-for-introduction} is a symbolic automaton that accepts
the sequences of integers in which
positive and negative integers occur alternately, and
the first element is a positive integer.
The \emph{symbolic automatic relations} introduced in this paper
 are relations recognized by symbolic automata that read input words (over
 a possibly infinite alphabet) in a synchronous manner.
 For example, consider the binary relation:
 \begin{align*}
 \REL_< = \{(l_{01}\cdots l_{0n}, l_{11}\cdots l_{1n})
   \in \setZ^*\times \setZ^*  %
   \mid l_{0i} < l_{1i} \text{ for every \(i\in\set{1,\ldots,n}\)}\}.
\end{align*}
 It is a symbolic automatic relation, as it is recognized by
 the symbolic synchronous automaton \( M_{<} \) on Fig.~\ref{fig:automata-for-introduction}
 (where \(l_0\) and \(l_1\) are bound to \(l_{0i}\) and \(l_{1i}\)
 at the \(i\)-th transition step).

 Symbolic automatic relations (SARs) inherit good properties of automatic relations
 and symbolic automata: SARs are closed under Boolean operations,
 and the emptiness problem of SAR (the problem of deciding whether \(\REL=\emptyset\),
 given the representation of a SAR \(\REL\))
 is decidable if the underlying theory (e.g. linear integer
 arithmetic) used for representing transitions is decidable.

 We are interested in the satisfiability problem for SARs,
 i.e., the problem of checking whether a given existentially-quantified formula
 (obtained by extending the signature of linear integer arithmetic
 with SARs and list constructors/destructors) is satisfiable.
 For example, whether \(\REL_<(X,Y) \land \REL_<(Y,X)\) is satisfiable
  (the answer is ``No'')
 is an instance of the problem.
 The class of existentially-quantified formulas considered in
 the satisfiability problem is reasonably expressive.
 For example,
 although the sortedness predicate \(\pred{sorted}(X)\) and the
 predicate \(\pred{nth}(X, i, x)\) (the \(i\)-th element of \(X\) is \(x\))
 are themselves not automatic relations, we can allow those predicates
 to occur in the formulas, as explained later.
 (Unfortunately, however, we cannot express the ``append'' relation.)

We first show that the satisfiability problem for SARs is undecidable, unfortunately.
The proof is based on a reduction from the undecidability of the halting problem for Minsky machines (or, two-counter machines).
Next, we show that the satisfiability problem for SARs can be reduced to
the satisfiability problem for Constrained Horn Clauses (CHCs) over integers
(\emph{without} lists). Thanks to the recent development of efficient CHC
solvers~\cite{KomuravelliGC16,ChampionCKS20,HojjatRummer18},
we can thus obtain a sound, automated (but incomplete) procedure for solving
the satisfiability problem for SARs.
We show, through experiments, that our reduction is effective, in that
the combination of our reduction with off-the-shelf CHC solvers can
solve the satisfiability problem for
various formulas over lists
that cannot be solved by state-of-the-art SMT solvers such as Z3 and CVC4.

Besides the above-mentioned improvement of
SMT solvers on recursive data structures, we also have in mind an application
to CHC solving
(indeed, improving CHC solvers was the original motivation of our work).
The goal of CHC solving is to check whether a given set of CHCs
has a model (interpretations for predicate variables
that make all the clauses valid).
Many of the CHC solvers prove the satisfiability of given CHCs
by constructing an actual model.
The main problem on such CHC solvers in dealing
with recursive data structures is that the language for describing models is too
restrictive: especially, it cannot express recursively defined predicates on
recursive data structures (apart from some built-in predicates such as the ``length'' predicate).
Our symbolic automatic relations can be used to enhance the expressive
power of the language.
The above-mentioned procedure for the SAR satisfiability problem can be
directly applied to an ICE-based CHC solver like \hoice~\cite{ChampionCKS20}.
\hoice consists of a learner, which constructs a candidate model, and a teacher, which checks whether the candidate is indeed a model.
Our procedure can be used by the teacher, when a given candidate is described
by using symbolic automatic relations.
Later in the paper, we give
examples of CHCs whose models can only be expressed by using
symbolic automatic relations, and show through experiments that
our procedure can indeed be used for checking the validity of models
described by using symbolic automatic relations.

Our contributions are: %
\begin{inparaenum}[(i)]
\item introduction of symbolic automatic relations and discussions of
  applications to SMT and CHC solving;
  \item a proof of the undecidability of the satisfiability problem on SARs;
  \item a sound (but incomplete) decision procedure for the satisfiability problem on SARs,
    via a reduction to CHC solving on integers
  \item an implementation and experiments to confirm the effectiveness of the above
    decision procedure.
\end{inparaenum}

The rest of this paper is structured as follows.
Section~\ref{sec:preliminaries} briefly reviews the notions used in many-sorted first-order logic.
Section~\ref{sec:sar} defines symbolic automatic relations and demonstrates how they can be used to express predicates over lists.
Section~\ref{sec:undecidabiltity} shows that the satisfiability problem for SARs is undecidable.
Section~\ref{sec:reduction} shows a reduction from the satisfiability problem for SARs to CHC solving, and Section~\ref{sec:experiments} reports experimental results.
Section~\ref{sec:related-work} discusses related work and Section~\ref{sec:conclusion} concludes the paper.

%% file: preliminaries.tex
\section{Preliminaries}
\label{sec:preliminaries}
This section introduces basic notions and notations used in the sequel.

\paragraph{Notations} Given a set \( S \), we write \( S^* \) for the set of all finite sequences over \( S \).
A word \( w \in S^* \) is either written as \( w = a_1 \cdots a_n \) or \( w = \ls{a_1, \ldots, a_n}\), where \( a_i \in S\); the empty word is denoted as \( \emptyword \) or \( \emptyls \).
The set of integers is written as \( \setZ \).

\subsection{FOL with list}
\subsubsection*{Syntax}
A (multi-sorted) \emph{signature} is a triple \( (\Ty , \fnsymb, \pdsymb)\), where \( \Ty\) is a set of sorts (aka types), \( \fnsymb \) is a set of (typed) function symbols and, \( \pdsymb \) is a set of (typed) predicate symbols.
There are two signatures that play important roles in this paper: the signature of integer arithmetic and the signature for integer lists.
We define the signature of integer arithmetic \(\sigint \) as \( ( \{ \intty \}, \fnint, \pdint) \).
The set \( \fnint \) contains the function symbols for integer arithmetic such as \( 0, s, + \) and \( \pdint \) contains the predicates for integer arithmetic such as \( {\eqint}, {<}, {\le } \); the precise definition is not important.
The signature of integer lists is defined by \( \siglint \defeq (\{\intty, \lintty \}, \fnint \cup \fnlint, \pdint  \cup \{ \eqlint \} ) \), where \( \fnlint \defeq \{\nil, \fsymb{cons}, \fsymb{head}, \fsymb{tail} \} \).
Here, \( \nil \) and \( \fsymb{cons} \) have type \( \lintty \) and \( (\intty, \lintty) \to \lintty \) and their intended meanings are the empty list and the ``cons function'', respectively.
As the name suggests, \( \fsymb{head} : \lintty \to \intty \) and \( \fsymb{tail} : \lintty \to \lintty \) will be interpreted as head and tail functions.

The set of \emph{terms} and the set of \emph{formulas} over a signature \( \tau  = (\Ty, \fnsymb, \pdsymb )\) is defined as follows:
\begin{align*}
    t &\Coloneqq x \mid f(t_1, \ldots, t_n) \\
    \phi &\Coloneqq \top \mid \bot \mid P(t_1, \ldots, t_n) \mid \lnot \phi \mid \phi_1 \wedge \phi_2 \mid \phi_1 \vee \phi_2 \mid \forall x \phi \mid \exists x \phi
\end{align*}
where \( x \) ranges over the denumerable set of variables, \( f \) ranges over \( \fnsymb\) and \( P \) ranges over \( \pdsymb \).
In what follows, we only consider well-typed formulas; we omit the definition of typing rules as they are standard.

Let us set some notational conventions for terms and formulas over the signature \( \siglint \) (or \( \sigint \)).
We use \( X, Y, Z, \ldots \) to range over the set of variables of type \( \lintty \) and \( x, y, z, \ldots \) to range over the set of variables of type \( \intty \).
A term is called an \emph{integer term} if it has type \( \intty \) and is called a \emph{list term} if it has type \( \lintty \).
We use \( T \) and \( t \) to range over the set of list terms  and integer terms, respectively.
We write \( \seq{x} \) (resp.\  \( \seq{X} \)) to represent a possibly empty sequence of integer variables (resp.\ list variables); \( \seq{t} \) and \( \seq{T} \) are defined similarly.

\subsubsection*{Semantics}
\newcommand{\model}{\mathcal{M}}
\newcommand{\mlint}{{\model_{\mathrm{list}}}}
\newcommand{\sort}{\iota}
A \emph{model} \( \mathcal{M} \) (or \emph{structure}) over a signature \( \signature = (\Ty, \fnsymb, \pdsymb)\) is a triple \( ((U_\sort)_{\sort \in \Ty}, (f^\model)_{f \in \fnsymb}, (P^\model)_{P \in \pdsymb})\), where each \( U_\sort \) is a non-empty set called a \emph{universe}, \( f^\model \) is a function over the universes, \( P^\model \) is a relation over the universes.
If \( f \) has type \( (\sort_1, \ldots, \sort_n ) \to \sort_{n + 1} \) then \( f^\model \) is a function from \(U_{\sort_1} \times \cdots \times U_{\sort_n} \) to \( U_{\sort_{n + 1}} \); the same applies to \( R^\model \).
We write \(( \setZ, (f^\setZ)_{f \in \fnint}, (P^\setZ)_{P \in \pdint} ) \) or simply \( \setZ \) for the standard model of integer arithmetic.
The \emph{standard model for integer lists} \( \model_{list} \) is a model over \( \siglint \) such that \( U_{\intty} \defeq \setZ \) and \( U_{\lintty} \defeq \setZ^* \); \( f^\mlint \defeq f^\setZ \) for \( f \in \fnint\);  \( P^\mlint \defeq P^\setZ \) for \( P \in \pdint\); \( \eqlint^\mlint \) is the diagonal relation on \( \setZ^* \); and the interpretations for symbols in \( \fnlint \) are defined in a natural way.
Formally, interpretations for symbols in \( \fnlint \) is defined by \( \nil^\mlint \defeq \emptyword \); \( \fsymb{cons}^\mlint(i, w) = i w\); \( \fsymb{head}^\mlint (i w ) \defeq  i \) and \( \fsymb{head}^\mlint (\emptyword) \defeq 0 \); and \( \fsymb{tail}^\mlint (i w) \defeq  w \) and \( \fsymb{tail}^\mlint (\emptyword) \defeq \emptyword \), where \( i \in \setZ\) and \( w \in \setZ^*\).\footnotemark
\footnotetext{Note that \( \fsymb{head}^\mlint \) and \( \fsymb{tail}^\mlint \) are defined as total functions.
This matches the behaviors of the existing SMT solvers such as Z3 or CVC4.
}

The semantics of terms and formulas are defined in the standard way.
Let \( \signature \) be a signature and \( \model \) be a model over \( \signature \).
An \emph{assignment} \( \alpha \) in \( \model \) maps variables of type \( \sort \) to elements of the universe associated with \( \sort \).
A triple \( (t, \model, \alpha ) \) of a term with type \( \sort \), a model and an assignment determines an element of the universe associated with \( \sort \), which we write as \( \sem{t}_{\model, \alpha}\).
Similarly, a triple \( (\phi , \model, \alpha) \), where \( \phi \) is a formula, determines whether the satisfaction relation \( \mathcal{M}, \alpha \models \phi \) holds.
We omit the precise definitions of \( \sem{t}_{\model, \alpha }\) and the satisfaction relation as they are defined as usual.
Since the truth or falsity of \( \model, \alpha \models \phi \) depends only on the values of \( \alpha \) for free variables of \( \phi \), we may write \( \model, [\seq{x} \mapsto \seq{a}] \models \phi \) if the free variables of \( \phi \) are among \( \seq{x} = x_1, \ldots, x_n\) and \(\seq{a} = \alpha(x_1), \ldots, \alpha(x_n)\).
We say that a formula \( \phi \) is \emph{satisfiable in \( \model \)} if there is an assignment \( \alpha\) such that \( \model, \alpha \models \phi \) and \( \phi \) is \emph{satisfiable} if there is a model \( \model \) in which \( \phi \) is satisfiable.
A formula \( \phi \) is \emph{valid in \( \model \)} if  \( \model, \alpha \models \phi\) holds for all assignments \( \alpha \); \( \phi \) is \emph{valid} if it is valid in all the models.
Two formulas \( \phi_1 \) and \( \phi_2\) are \emph{\( \mathcal{M} \)-equivalent} if, for all assignments \( \alpha \), \( \mathcal{M}, \alpha \models \phi_1\) if and only if \( \mathcal{M}, \alpha \models \phi_2 \).

%% file: sar.tex
\section{Symbolic Automatic Relations}\label{sec:sar}
In this section, we introduce the notion of \emph{symbolic automatic relations}.
We first introduce the notion of a \emph{symbolic synchronous automaton}
in Section~\ref{sec:ssa}, which is a special kind of symbolic automaton~\cite{DBLP:conf/cav/DAntoniV17}, which serves as the representation of a symbolic automatic relation.
We then define symbolic automatic relations in Section~\ref{subsec:sar}.
For the sake of simplicity we consider symbolic automatic relations on
integer sequences (or, lists of integers).
It would not be
difficult to extend them to deal with
tree-structured data, by using (symbolic, synchronous)
tree automata; see also Remark~\ref{rem:trees}.

\subsection{Symbolic Synchronous Automata}
\label{sec:ssa}

We first extend the model \( \setZ \) by adding
the special padding symbol \( \pad \),
which will be used in the definition of symbolic synchronous automata.
\begin{definition}[Partial model for integer arithmetic]
A \emph{partial model for integer arithmetic} \( (\Zp, (f^{\Zp})_{f \in \fnint}, (P^{\Zp})_{P \in \pdint} )\) is a model over the signature \( \signature_{\mathrm{int} \cup \{ \pad \}} \defeq  (\{ \intty\}, \fnint, \pdint \cup \{ \ispad \})\), where
\begin{itemize}
\item the universe \( \Zp \) is \(\setZ \cup \{ \pad \}\), where
   \( \pad \notin \setZ \) is called a padding symbol,
    \item for every \( k \)-ary function symbol \( f \), \( f^{\Zp}(a_1, \ldots, a_k) \defeq f^{\setZ}(a_1, \ldots, a_k) \) if \( a_i \in \setZ \) for all \(1 \le i \le k\) and \( f^{\Zp}(a_1, \ldots, a_k) \defeq \pad \) otherwise, and
    \item for every \( P \in \pdint \), \( P^{\Zp} \defeq P^{\setZ} \) and \( \ispad^{\Zp} \defeq \{ \pad \}\).
\end{itemize}
By abuse of notation, we may write \( \Zp \) to denote the partial model for integers.
\end{definition}

\begin{remark}
  The semantics of the negation \(\neg\)
  is a little tricky for %
  the partial model.
  For example, the interpretation of \( x < y \) is different from \( \lnot (x \ge y) \):
  ``\(1 < \pad\)'' is false but ``\(\neg(1\ge \pad)\)'' is true.
\end{remark}

\begin{definition}[Symbolic synchronous automaton]
A \emph{\( k \)-ary symbolic synchronous nondeterministic finite automaton} with \( n \) parameters\footnotemark\  \(\seq{x} = x_0, \ldots, x_{n - 1} \) (\((k, n)\)-ary ss-NFA for short) is a quadruple \( M(\seq{x}) =(Q,I,F,\Delta) \) where
\footnotetext{The parameters \( \seq{x} \) are ``bound variables'' and we identify ``\( \alpha \)-equivalent'' automata.}
\begin{itemize}
  \item \( Q \) is a finite set of \emph{states},
  \item \( I \subseteq Q \) is the set of \emph{initial states},
  \item \( F \subseteq Q \) is the set of \emph{final states},
  \item \( \Delta \subseteq Q \times \Psi^{\seq{x}}_k \times Q \) is a finite set of \emph{transitions}.
  Here \( \Psi^{\seq{x}}_k\) is a subset of formulas over the signature \( \signature_{\mathrm{int} \cup \{ \pad \}} \) containing only formulas, whose free variables are among \( l_0, \ldots, l_{k - 1}, x_0, \ldots, x_{n - 1} \).
\end{itemize}
\skchanged{Intuitively, the free variables \( l_i \) are bound to the \( i \)-th input at each transition step and \( x_i \) are bound to integer values that do not change at each step.}
Formally, for \(\seq{j} = (j_0, \ldots, j_{n - 1}) \in \setZ^n \) and \( a = (i_0, \ldots, i_{k - 1}) \in \Zp^k \), an \emph{\( a \)-transition} of \( M(\seq{j}) \) is a transition \( q \trans{\varphi} q' \) such that \( \Zp, [\seq{x} \mapsto \seq{j}, l_0 \mapsto i_0, \ldots, l_{k - 1} \mapsto i_{k - 1}] \models \phi \).
This \( a \)-transition is denoted as \(q \trans{a}_{M(\seq{j})} q' \) (or \( q \trans{a} q' \) when \( M(\seq{j}) \) is clear from the context).

A \( (k, n) \)-ary ss-NFA \( M(\seq{x}) = (Q, I, F, \Delta) \) is \emph{effective} if
\begin{align*}
\{ (i_0, \ldots, i_{k -1}, j_0, \ldots, j_{n - 1}) \in \Zp^k \times \setZ^n \mid \Zp, [\seq{l} \mapsto \seq{i}, \seq{x} \mapsto \seq{j}] \models \phi \}
\end{align*}
is a decidable set for all \(q \trans{\phi}q' \in \Delta \).
We sometimes call a \((k,0)\)-ary ss-NFA just a \(k\)-ary ss-NFA.
\end{definition}

\skchanged{
The existence of parameters allows us to use ss-NFAs as representations of relations that take not only words but also integers as arguments.
}
\begin{definition}[Language of ss-NFA]
Let \( M(\seq{x}) = (Q, I, F, \Delta) \) be a \( (k, n) \)-ary ss-NFA and \( \seq{j}  \in \setZ^n \).
A word \( w \in (\Zp^k)^* \) is \emph{accepted} by \( M(\seq{j}) \) if
\begin{itemize}
  \item \( w = \emptyword \) and \( I \cap F \neq \emptyset \), or
  \item \( w = a_1 \cdots a_m \) and for all \( 1 \le i \le m \), there exist transitions \( q_i \trans{a_i}_{M(\seq{j})} q_{i+1} \) such that \( q_1\in I \) and \( q_{m+1} \in F \).
\end{itemize}
A word accepted by \( M(\seq{j}) \) is called an \emph{accepting run of \( M(\seq{j}) \)}.
The \emph{language accepted by \( M(\seq{j}) \)}, denoted \( \lang(M(\seq{j})) \), is the set of words accepted by \( M(\seq{j}) \).
We also write \( \lang(M(\seq{x})) \) for the relation \( \{ (w, \seq{j}) \mid w \in \lang(M(\seq{j})) \} \).
\end{definition}

\begin{example}
\label{ex:automata-for-sorted-and-nth}
Consider the ss-NFAs in Fig.~\ref{fig:automata-for-running-example}.
The automaton \( M_1 \) is formally defined as a 2-ary ss-NFA without parameter \( M_1 \defeq (\{ q \}, \{ q \}, \{ q \}, \{(q, \lnot(l_0 > l_1), q) \} ))  \).
The automaton \( M_2(x) \) is a 3-ary ss-NFA, with one parameter \( x \), defined by \( M_2(x) \defeq (\{ q_0, q_1 \}, \{ q_0\}, \{q_1\},  \Delta ) \), where \( \Delta \defeq \{ (q_0, l_0 = l_1 + 1, q_0), (q_0, l_0 = 0 \wedge l_2 = x ,q_1), (q_1, \top, q_1) \} \); the automaton \( M_3(x) \) can be formally described in a similar manner.
These automata will be used to define predicates \( \pred{sorted}(X) \) and \( \pred{nth}(i, x, X)\) later in this section.

\skchanged{The acceptance language of \( M_2(x)\), i.e.~\( \lang(M_2(x)) \subseteq  \Zp^3 \times \setZ \), is given as}
\begin{align*}
  \left\{ \left(
     \left[
        \begin{array}{c}
            a_{01} \\
            a_{11} \\
            a_{21}
        \end{array}
    \right]
  \cdots
  \left[
    \begin{array}{c}
      a_{0n} \\
      a_{1n} \\
      a_{2n}
    \end{array}
  \right]
  , j \right) \, \middle| \,
  \begin{aligned}
   \exists m.\ &1 \le m \le n \land  a_{0m} = 0 \land  a_{2m} = j   \\
  &\land \ (\forall i .\ 1 \le i < m \implies a_{0i} = a_{1i} + 1)
  \end{aligned}
\right\}.
\end{align*}
\qed
\end{example}

 \begin{figure}[t]
 \begin{minipage}{0.18\linewidth}
 \fbox{\( M_1 \)}
 \par
 \centering
 \begin{tikzpicture}
 \node[state, initial, accepting] (q) {\(q\)};
 \draw (q) edge[loop above] node{\( \lnot (l_0 > l_1) \)} (q);
 \end{tikzpicture}
 \end{minipage}
\begin{minipage}{0.38\linewidth}
\fbox{\( M_2(x)\)}
 \par
 \centering
\begin{tikzpicture}
\node[state, initial] (q0) {\( q_0 \)};
\node[state, accepting,right of=q0, xshift=1cm] (q1) {\( q_1 \)};
\draw (q0) edge[loop above] node{\( l_0 = l_1 + 1 \)} (q0)
       (q0) edge[above] node{\(l_0 = 0 \land  l_2 = x \)} (q1)
       (q1) edge[loop above] (q1);
\end{tikzpicture}
\end{minipage}
\begin{minipage}{0.42\linewidth}
\fbox{\( M_3(x) \)}
\par
\begin{tikzpicture}
\node[state, initial] (q0) {\( q_0 \)};
\node[state, accepting,right of=q0, xshift=1.2cm] (q1) {\( q_1 \)};
\node (label) at (1.6, -0.8) {\( \vee \ l_0 < 0 \)}; %
\draw (q0) edge[loop above] node{\( l_0 > 0 \wedge l_0 = l_1 + 1 \)} (q0)
      (q0) edge[below] node{\((l_0 = 0 \land  l_2 \neq x)\)} (q1)
      (q1) edge[loop above] (q1);
\end{tikzpicture}
\end{minipage}
\caption{Examples of ss-NFAs. }
\label{fig:automata-for-running-example}
\end{figure}
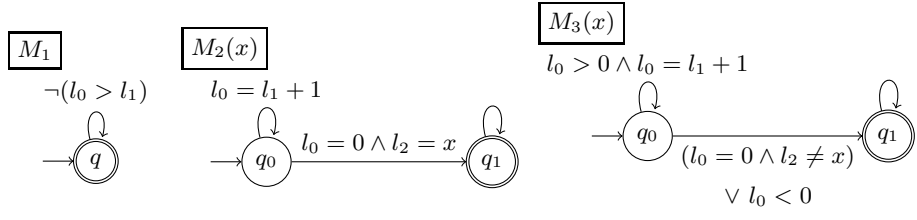

 We introduce some terminology on %
 ss-NFAs.
A \( (k, n) \)-ary ss-NFA \( M(\seq{x}) = (Q, I, F, \Delta)\) is \emph{deterministic} if \( |I|=1 \) and for all transitions  \( q \trans{\phi_1} q_1 \) and \(  q \trans{\phi_2} q_2 \), if \( \varphi_1\land\varphi_2 \) is satisfiable in \( \Zp \) then \( q_1 = q_2 \).
A state \( q \) of \( M(\seq{x}) \) is called \emph{complete} if for all \( a \in  \Zp^{k} \) and \( \seq{j} \in \setZ^n \) there exists an \( a \)-transition \( q \trans{a}_{M ( \seq{j})} q'\) for some \( q' \).
A ss-NFA \( M(\seq{x}) \) is \emph{complete} if all states of \( M(\seq{x}) \) are complete.

Since ss-NFA are just a special kind of
symbolic automata and symbolic automata can be determinized and completed, it can be shown that ss-NFAs are closed under Boolean operations using variants of the complement construction and the product construction of standard automata.\footnotemark
\footnotetext{The fact that determinization is possible and that symbolic automata are closed under boolean operations were originally shown for symbolic automata without parameters~\cite{VeanesHT10}, but the existence of parameters does not affect the proof.}
\begin{proposition}[Closure under boolean operations~\cite{VeanesHT10}]
\label{prop:snfa-closed-bolean}
Given \( (k, n) \)-ary ss-NFAs \( M_1(\seq{x}) \) and \( M_2(\seq{x}) \), one can effectively construct ss-NFAs \( M_1^c(\seq{x}) \) and \( (M_1\times M_2)(\seq{x}) \) such that \( \lang(M_1^c(\seq{x}))= (( \Zp^k )^* \times \setZ^n) \setminus \lang(M_1(\seq{x})) \) and \( \lang((M_1\times M_2)(\seq{x})) = \lang(M_1(\seq{x}))\cap \lang(M_2(\seq{x})) \).
Moreover, if \( M_1(\seq{x}) \) and \( M_2(\seq{x}) \) are effective, so are \( M_1^c(\seq{x}) \) and \( (M_1 \times M_2)(\seq{x}) \).
\qed
\end{proposition}

\subsection{Symbolic Automatic Relations}
\label{subsec:sar}

A symbolic automatic relation (SAR) is basically an acceptance language of a ss-NFA, but not every acceptance language of a ss-NFA is a SAR.
Recall that a run of a \( k \)-ary ss-NFA is a word \( w \in (\Zp^k)^* \) and thus it does not necessarily correspond to tuples of words over \( \setZ \) since the ``padding symbol'' can appear at any position of \( w \).
In order to exclude such ``invalid inputs'', we first define the convolution operation, which converts a tuple of words to a word of tuples.

\begin{definition}[Convolution]
\label{def:convolution}
Given \( k \) words \( w_0,...,w_{k - 1} \in \setZ^* \), with \( w_i = a_{i1}\cdots a_{i l_i} \) and \( l=\max(l_0,...,l_{k - 1}) \), the \emph{convolution of words} \( w_0,...,w_{k - 1} \), denoted as \( \conv(w_0,...,w_{k -1}) \), is defined by
\begin{align*}
    \conv(w_0,...,w_{k - 1}) \defeq
    \left[
        \begin{array}{c}
            a_{01}' \\
            \vdots \\
            a_{(k - 1)1}'
        \end{array}
    \right]
        \cdots
    \left[
        \begin{array}{c}
            a_{0l}' \\
            \vdots \\
            a_{(k -1)l}'
        \end{array}
    \right] \in \left(\Zp^k\right)^*
    \quad \text{ and } \quad \conv() \defeq \emptyword
\end{align*}
where \( a_{ij}' = a_{ij} \) if \( j \le l_i \) and \( a_{ij}' = \pad \) otherwise.
The padding symbol is appended to the end of some words \( w_i \) to make sure that all words have the same length.
\end{definition}

We write \( \convZ{k} \) for the set of convoluted words, i.e.~\( \convZ{k} \defeq \{\conv(w_0, \ldots, w_{k - 1}) \mid (w_0, \ldots, w_{k - 1}) \in (\setZ^*)^k \}\).
This set can be recognized by a ss-NFA.
\begin{proposition}
\label{prop:set-of-convoluted-words}
Let \( k \) and \( n \) be natural numbers.
Then there is a \( (k, n) \)-ary ss-NFA \( M(\seq{x}) \) such that \( \lang(M(\seq{x})) = \convZ{k} \times \setZ^n \).
\qed
\end{proposition}
Because of this proposition, there is not much difference between ss-NFAs that only take convoluted words as inputs and ss-NFAs that take any word \( w \in (\Zp^k)^* \) as inputs.
Given a ss-NFA \( M(\seq{x})\), we can always restrict the form of inputs by taking the product with the automaton that recognizes \( \convZ{k} \times \setZ^n \).

\begin{definition}[Symbolic automatic relation]

A relation \( \relR \subseteq (\setZ^*)^k \times \setZ^n \) is called a \emph{ \( (k, n) \)-ary symbolic automatic relation} (SAR) if \( \{(\conv(w_0, \ldots, w_{k - 1}), \seq{j}) \mid (w_0, \ldots, w_{k -1}, \seq{j}) \in \relR \}  = \lang(M(\seq{x})) \) for some \( (k, n) \)-ary ss-NFA \( M(\seq{x}) \); in this case, we say that \( \relR \) is recognized by \( M(\seq{x}) \).

Given a \( (k, n)\)-ary ss-NFA \( M(\seq{x}) \), the \emph{\( (k, n) \)-ary SAR represented by \( M(\seq{x})\)}, denoted as \( \relR(M(\seq{x})) \), is defined as \( \{ (w_0, \ldots, w_{k - 1}, \seq{j} ) \mid (\conv(w_0, \ldots, w_{k - 1}), \seq{j}) \in \lang(M(\seq{x})) \}\).
Note that \( \relR(M(\seq{x})) \) is indeed a SAR because \( \relR(M(\seq{x})) \) is recognized by the product of \( M(\seq{x}) \) and the automaton that recognizes \(\convZ{k} \times \setZ^n \).
\end{definition}

\subsection{Expressing Predicates on Lists}
\label{sec:sar:predicate}
We demonstrate that various predicates over lists can be expressed as logical formulas obtained by extending the signature \( \siglint \) with SARs.
Moreover, we show that those predicates belong to
a class of formulas called \emph{\( \sigmaosar \)-formulas}.
We are interested in \( \sigmaosar \)-formulas because, as we shall see in Section~\ref{sec:experiments}, checking whether a simple \( \sigmaosar \)-formula is satisfiable can often be done automatically.

Henceforth, we allow SARs to appear in the syntax of formulas.
Formally, we consider formulas over \( \sigsar\), where \( \sigsar \) is defined as
the signature obtained by adding predicate symbols of the form \( R_{M(\seq{x})}\), which we also call SAR, to the signature \( \siglint \).
Here the subscript \( M(\seq{x}) \) represents ss-NFAs.
In what follows, the term ``formula'' means
a formula over the signature \( \sigsar \), unless the signature is explicitly specified.
The predicate symbols \( R_{M(\seq{x})} \) are interpreted symbols.
We consider a fixed model \( \mathcal{M} \) in which every predicate symbol of the form \( R_{M(\seq{x})} \) is interpreted as the symbolic automatic relation represented by \( M(\seq{x}) \) and other symbols are interpreted as in \( \Mlint \).

\begin{definition}
  A formula \( \phi \) is a \emph{\( \deltazsar \)-formula} if
  one can effectively construct
  a formula \( R_{M(\seq{x})}(\seq{T}, \seq{t})\) (where \( R_{M(\seq{x})}\) is a SAR)
  that is \( \model \)-equivalent to \( \phi \).
A formula \( \phi \) is a \emph{\( \sigmaosar \)-formula} if one can effectively construct a formula of the form \( \exists \seq{x} \exists \seq{X} \phi_0\) that is \( \model \)-equivalent to \( \phi \) and \( \phi_0 \) is a \( \deltazsar \)-formula.
We say that a formula \( \phi \) is a \emph{\( \deltaosar \)-formula} if both \( \phi \) and \( \lnot \phi \) are \( \sigmaosar \)-formulas.
\end{definition}

\begin{example}
\label{ex:sotrted-and-nth}
Let us consider the predicate \( \pred{sorted}(X) \), which holds
just if \( X \) is sorted in ascending order.
The predicate \( \pred{sorted} \) can be  defined
as a \( \deltazsar \)-formula,
by \( \pred{sorted}(X) \defeq R_{M_1}(X, \tail{X}) \),
where \( M_1 \) is the ss-NFA used in Example~\ref{ex:automata-for-sorted-and-nth}.

The predicate \( \pred{nth}(i, x, X) \), meaning that ``the \( i \)-th element of \( X \) is \( x \)'', can be defined as a \( \deltaosar \)-formula.
To show this, we use the automata \( M_2 \) and \( M_3 \) used in Example~\ref{ex:automata-for-sorted-and-nth}.
We can define \( \pred{nth} \) by \( \pred{nth}(i, x, X) \defeq \exists Y.\ R_{M_2(x)}(\ins{i}{Y}, Y, X, x) \).
In this definition, the list represented by \( Y \) works as a ``counter''.
Suppose that \( \ins{i}{Y}\) is interpreted as \( w_0 \) and assume that the first \( n \) transitions were all \(q_0 \trans{l_0 = l_1 + 1} q_0\).
Then we know that the list \( w_0 \) must be of the form \( [i, (i - 1),  \ldots (i - n), \ldots]  \), which can be seen as a decrementing counter starting from \( i \).
The transition from \( q_0 \) to \( q_1 \) is only possible when the counter is \( 0 \) and this allows us to ``access to the \( i \)-th element'' of the list represented by \( X \).
The negation of \( \pred{nth}(i, x, X) \) can be defined as \(\exists Y. R_{M_3(x)}(\ins{i}{Y}, Y, X, x) \).
Using the same technique, we can define the predicate \( \pred{length}(X, i) \) (``the length of \( X \) is \( i \)'') as a \( \deltaosar \)-formula.
\qed
\end{example}

The following proposition and example are useful for constructing new examples of \( \sigmaosar \)-formulas.

\begin{proposition}
\label{prop:deltao-deltaz-boolean-closedness}
The class of \( \deltazsar \)-formulas and \( \deltaosar \)-formulas are closed under boolean operations.
\end{proposition}
\begin{proof}
The fact that \( \deltazsar \)-formulas are closed under boolean operations follows from the fact that ss-NFAs are closed under boolean operations (Prop.~\ref{prop:snfa-closed-bolean}) and that the set of convoluted words is a language accepted by a ss-NFA (Prop.~\ref{prop:set-of-convoluted-words}).

By the definition of \( \deltaosar \)-formulas, \( \deltaosar \)-formulas are clearly closed under negation.
Given formulas \( \exists \seq{x} \exists \seq{X} \phi_1\) and \( \exists \seq{y} \exists \seq{Y} \phi_2\), where \( \phi_1 \) and \( \phi_2 \) are \( \deltazsar \)-formulas, \( (\exists \seq{x} \exists \seq{X} \phi_1) \wedge ( \exists \seq{y} \exists \seq{Y} \phi_2) \) is equivalent to \( \exists \seq{x}\seq{y} \exists \seq{X} \seq{Y}( \phi_1 \wedge \phi_2)\).
Since \( \phi_1 \wedge \phi_2 \) is a \( \deltazsar \)-formula, it follows that \( \deltaosar \)-formulas are closed under conjunction.
\qed
\end{proof}

\begin{example}
\label{ex:arith-and-eq-are-delta0}
Every arithmetic formula, i.e.~a formula over the signature \( \sigint \), is a \( \deltazsar \)-formula.
Given an arithmetic formula \( \phi \) whose free variables are \( x_0, \ldots, x_{k - 1}\), we can construct a \( \model \)-equivalent formula \( R_{M_{\phi}}([x_0], \ldots, [x_{k - 1 }]) \), where \( [x_i] \defeq \ins{x_i}{\nil}\) and \( M_{\phi} \defeq ( \{ q_0, q_1\}, \{q_0 \}, \{ q_1 \}, \{ (q_0, \phi[l_0/x_0, \ldots, l_{k - 1}/x_{k-1}], q_1) \}) \).
Similar transformation works even if a formula of the form \( \hd{X} \) appears inside a formula \( \phi \) of type \( \intty \).

Equality relation on two lists is also a \( \deltazsar \)-formula because it can be described by a ss-NFA. \qed
\end{example}

Thanks to Proposition~\ref{prop:deltao-deltaz-boolean-closedness} and Example~\ref{ex:arith-and-eq-are-delta0}, we can now write various specification over lists as (negations of) closed \( \sigmaosar \)-formulas.
For example, consider the following formula that informally means ``if the head element of a list sorted in ascending order is greater or equal to 0, then all the elements of that list is greater or equal to 0'':
\begin{align*}
\phi \defeq\forall x.\  \forall i.\  \forall X.\ X \neq \nil \wedge \pred{sorted}(X) \wedge \hd{X} \ge 0 \wedge \pred{nth}(i, x, X) \implies x \ge 0
\end{align*}
The negation of \( \phi \) is a  \( \sigmaosar \)-formula because \( \pred{sorted}(X) \) and \( \pred{nth}(i, x, X)\) are \( \deltaosar \)-formulas as we saw in Example~\ref{ex:sotrted-and-nth}.
Note that the validity of \( \phi \) can be checked by checking that \( \lnot \phi \) is unsatisfiable, which can be done by a satisfiability solver for \( \sigmaosar \)-formulas.

\subsection{An Application to ICE-Learning-Based CHC Solving with Lists}
\label{sec:application-to-CHC}
We now briefly discuss how a satisfiability solver for \( \sigmaosar \)-formulas may be used in the teacher part of ICE-learning-based CHC solvers for lists.
As mentioned in Section~\ref{sec:intro},
ICE-learning-based CHC solvers~\cite{ChampionCKS20,DBLP:journals/pacmpl/EzudheenND0M18} consist of a learner, which constructs a candidate model, and a teacher, which checks whether the candidate is indeed a model, i.e., whether the candidate model satisfies every clause.
Each clause is of the form
\begin{align*}
\forall \seq{x}.\, \forall \seq{X} . A_1\land \cdots \land A_n \implies B
\end{align*}
where \(A_1,\ldots,A_n\) and \(B\) are either primitive constraints or
atoms of the form \( P(t_1,\ldots,t_k) \) where \( P \) is a predicate variable.
Assuming that the learner returns an assignment \( \theta \) of
\(\deltaosar\)-formulas to predicate variables,
the task of the teacher is to check that 
\begin{align*}
\varphi := \forall \seq{x}.\, \forall \seq{X}. \theta A_1\land \cdots \land \theta A_n \implies \theta B
\end{align*}
is a valid formula for each clause
\(\forall \seq{x}.\, \forall \seq{X}. \theta A_1\land \cdots \land \theta A_n \implies \theta B\).
The negation of \(\varphi\) can be expressed as a closed \(\sigmaosar\)-formula
\(\exists \seq{x}.\, \exists \seq{X}. R_{M(\seq{y})}(\seq{T}, \seq{t})\).
By invoking a satisfiability solver for a \(\sigmaosar\)-formulas we can check whether \( R_{M(\seq{y})}(\seq{T}, \seq{t}) \) is unsatisfiable, which is equivalent to checking if \( \phi \) is valid.
If \( R_{M(\seq{y})}(\seq{T}, \seq{t}) \) is satisfiable, then \( \varphi \) is invalid.
In this case, the teacher should generate a counterexample against \(\varphi\).

\begin{example}
\label{ex:chc-for-sorted}
Let us consider the following set of constrained horn clauses:
\begin{align*}
\begin{array}{l}
    P(X) \Leftarrow X = nil \lor X = \fsymb{cons}(x,nil) \\
    P(\fsymb{cons}(x,\fsymb{cons}(y,X))) \Leftarrow x \leq y \land P(\fsymb{cons}(y,X)) \\
    Q(\fsymb{cons}(0,X)) \Leftarrow \top \quad\;
    Q(\fsymb{cons}(x,X)) \Leftarrow Q(X) \quad\;
    \bot \Leftarrow P(\fsymb{cons}(1,X)) \land Q(X)
\end{array}
\end{align*}
A model of this set of CHCs is \( P(X) \mapsto \pred{sorted}(X), Q(X) \mapsto \pred{hasZero}(X) \), where \( \pred{sorted} \) is the predicate we have seen in Example~\ref{ex:sotrted-and-nth} and \( \pred{hasZero}(X) \) is a predicate that holds if \( 0 \) appears in the list \( X \).
It is easy to check that \( \pred{hasZero}\) is a \( \deltazsar \)-formula.%
Hence, if there is a learner part, which is yet to be implemented, that provides \( \pred{sorted} \) and \( \pred{hasZero} \) as a candidate model, then a satisfiability solver for \( \sigmaosar \)-formulas can check that this candidate model is a valid model.
\qed
\end{example}

\begin{remark}
  As discussed later in Section~\ref{sec:reduction}, the satisfiability problem for
  SARs is solved by a reduction to CHC solving without data structures.
  Thus, combined with the translation above,
  we translate a part of the problem of solving CHCs with data structures
  to the problem of solving CHCs without data structures.
  This makes sense because solving CHCs without data structures is
  easier \emph{in practice} (although the problem is undecidable in general, even
  without data structures). One may wonder why we do not directly translate
  CHCs with data structures to those without data structures, as in
  \cite{DBLP:journals/tplp/AngelisFPP18a,DeAngelisFPP20}. The detour through
  SARs has the following advantages. First, it provides a uniform, streamlined
  approach thanks to the closure of SARs under Boolean operations. Second,
  SARs serve as certificates of the satisfiability of CHCs.
\end{remark}

%% file: undecidability.tex
\section{Undecidability Result}
\label{sec:undecidabiltity}
This section shows that the satisfiability problem for SARs is undecidable in general.
The \emph{satisfiability problem for SARs} is the problem of deciding whether there is an assignment \( \alpha \) such that \( \model, \alpha \models \phi \), given a \( \sigsar \)-formula \( \phi \).
\changed{We prove that the satisfiability problem is undecidable even for the class of \(\deltazsar\)-formulas,}
by reduction from the halting problem for two-counter machines.

\newcommand{\Line}{\mathbf{Line}}
\newcommand{\Code}{\mathbf{Code}}
\newcommand{\Inst}{\mathbf{Inst}}
\newcommand{\inc}[2]{\mathbf{inc}(#1, #2)}
\newcommand{\jzdec}[3]{\mathbf{jzdec}(#1, #2, #3)}
\newcommand{\halt}{\mathbf{halt}}

\begin{definition}[Minsky's two-counter machine~\cite{Minsky61}]
\emph{Minsky's two-counter machine} consists of (i) two integer registers \( r_0 \) and \( r_1 \), (ii) a set of instructions, and (iii) a program counter that holds a non-negative integer.
Intuitively, the value of the program counter corresponds to the line number of the program currently being executed.

A \emph{program} is a pair \( P = (\Line, \Code) \), where \( \Line \) is a finite set of non-negative integers such that \( 0 \in \Line \) and \( \Code \) is a map from a finite set of non-negative integers to \( \Inst \), the set of \emph{instructions}.
The set \( \Inst \) consists of:
\begin{itemize}
  \item \( \inc{i}{j} \) : Increment the value of register \( r_i \) and set the program counter to \( j \in \Line \).
  \item \( \jzdec{i}{j}{k} \) : If the value of register \( r_i \) is positive, decrement it and set the program counter to \( j \in \Line \). Otherwise, set the program counter to  \( k \in \Line\).
  \item \( \halt \) : Stop operating.
\end{itemize}
Initially, \( r_0 \), \( r_1 \) and the program counter are all set to \( 0 \).
Given a program \( P = (\Line, \Code)\) the machine executes \( \Code(i) \), where \( i \) is the value of the program counter until it executes the instruction \( \halt \).
\end{definition}

Given a program \( P \) we simulate its execution using a formula of the form \( R_{M}(X_0, X_1, \tail{X_0}, \tail{X_1})\).
The states and edges of \( M \) are used to model the control structure of the program, and the list variables \( X_0 \) and \( X_1 \) are used to model the ``execution log'' of registers \( r_0 \) and \( r_1 \), respectively. 

\begin{theorem}[Undecidability of satisfiability of SARs]
  Given a \( k \)-ary symbolic automatic relation \( R_M \) represented by an effective \( (k,0) \)-ary ss-NFA \( M \) and list terms \( T_1, \ldots ,T_k \),
\changed{it is undecidable whether \( R_M(T_1, \ldots,T_k) \) is satisfiable.}
\end{theorem}
\begin{proof}
\newcommand{\final}{q_{\mathrm{accept}}}
We show that for a given program \( P  = (\Line, \Code )\), we can effectively construct a SAR \( R_{M_P} \) that satisfies ``\( P \) halts iff there are assignments for \(X_0 \)  and \( X_1 \) that satisfy \( R_{M_P}(X_0,X_1,\tail{X_0},\tail{X_1})\)''.
\changed{Intuitively, \(X_i\) denotes the ``execution log of \(r_i\)'',
i.e., the sequence of values taken by \(r_i\) in a terminating execution sequence
of \(P\), and \( R_{M_P}\) takes as arguments
both \(X_i\) and \(\tail{X_i}\) to check that \(X_i\) represents a valid sequence.}
The ss-NFA \( M_P \) is defined as \( (Q,I,F,\Delta) \), where \( Q \defeq \{q_i \mid i \in \Line \} \cup \{ \final \}\), \( I \defeq \{q_0\} \), and \( F \defeq \{ \final \} \).
We define the set of transitions
\changed{so that \(M_P\) has a transition \( q_i \trans{(c_0,c_1,c'_0,c'_1)}q_j \) iff
the two-counter machine has a transition from the configuration \((i,c_0,c_1)\)
(where \(i\) is the current program pointer and \(c_i\) is the value of \(r_i\)) to
\((j,c'_0,c'_1)\).}
We also add transition from \( q_i \) to the final state \( \final \) if \( \Code(i) = \halt\).
Formally, \( \Delta \) is defined as the smallest set that satisfies the following conditions:
\begin{itemize}
    \item \( (q_i,l'_{r} = l_r + 1 \land l_{1-r} = l_{1-r},q_j)\in \Delta \) if  \(\Code(i) = \inc{r}{j} \).
    \item \( (q_i,l_r > 0 \land l'_{r} = l_r - 1 \land l'_{1-r} = l_{1-r},q_j),(q_i,l_r = 0 \land l'_0 = l_0 \land l'_1 = l_1,q_k)\in\Delta \) if \( \Code(i) = \jzdec{r}{j}{k} \).
    \item \( (q_i, \ispad(l'_0) \land \ispad(l'_1),q_{accept})\in\Delta \) if \( \Code(i) = \halt \).
\end{itemize}
Here, we have written \(l'_0\) and \(l'_1\) for \(l_2\) and \(l_3\).

Based on the intuitions above, it should be clear that 
\( P \) halts and the execution log of \( r_i \) obtained by running \( P \) is \( w_i \),
if and only if
 \( M_P \) accepts \( \conv(w_0,w_1, w'_0, w'_1) \), where \( w'_i\) is the tail of \( w_i\).
Thus, \(P\) halts if and only if \(R_{M_P}\) is satisfiable.
Since the halting problem for two-counter machines is undecidable, so is the satisfiability of
\(R_{M}\).
\comment{
 Our goal is to show that ``\( P \) halts and the execution log of \( r_i \) obtained by running \( P \) is \( w_i \)'' \( \Leftrightarrow \) ``\( M_P \) accepts \( \conv(w_0,w_1, w'_0, w'_1) \), where \( w'_i\) is the tail of \( w_i\)''.
 We only show the left-to-right direction; the opposite direction can be proved similarly.

($\Rightarrow$)
Suppose that \( P \) halted after executing the sequence of instructions \( \Code(i_0), \Code(i_1), \ldots, \Code(i_n)\), where \( i_0 = 0\) and \(\Code(i_n) = \halt \).
Moreover, assume that, at the \( l \)-th step of the execution, the value set to \( r_0 \) and \( r_1 \)  were \( j_{0l}\) and \(j_{1l}\), respectively.
Then define \( w_r \defeq j_{r0} \cdots j_{rn} \) and \( w'_r \defeq j_{r1} \cdots j_{rn}\) for \( r \in \{ 0, 1 \}\).
We show that \( \conv(w_0, w_1, w'_0, w'_1)\) is an accepting run.
Let \(0 \le m \le l-1\).
Then \( \Code(i_m) \neq \halt \) and, because of the way \( \Delta \) is constructed, there exist the transitions \( q_{i_m} \trans{(j_{0m},j_{1m},j_{0(m+1)},j_{1(m+1)})}_{M_P} q_{i_{m+1}} \).
If \( m = l \), then \( \Code(i_m) = \halt \).
So there is a transition \(q_{i_m} \trans{(j_{0m}, j_{1m} \pad, \pad)} \final \).

Since the halting problem for two-counter machine is undecidable, the satisfiability problem for symbolic automatic relations is undecidable.
}
\qed
\end{proof}

%% file: reduction.tex
\section{Reduction to CHC Solving}
\label{sec:reduction}
This section describes the reduction from the satisfiability problem of SARs to CHC solving, whose constraint theory is mere integer arithmetic.
Precisely speaking, the reduction works for a fragment of \( \sigsar \)-formulas, namely the \( \sigmaosar \)-formulas.
This section starts with a brief overview of the reduction.
We then give the formal definition and prove the correctness of the reduction.

\subsection{Overview}
Let us first present intuitions behind our reduction using an example.
{%
\newcommand{\qformula}{\exists x \exists X .\pred{nth}(i, x, X)}
\newcommand{\formula}{\pred{nth}(i, x, X)}
Consider the predicate \( \pred{nth} \) we defined in Example~\ref{ex:sotrted-and-nth}.
Let \(i\) be a non-negative integer constant, and
suppose that we wish to check whether \(\formula\) is satisfiable,
i.e., whether \( \qformula \) holds.
We construct a set of CHCs \( \Pi \) such that \( \model \models \qformula \)
iff \( \Pi \) is unsatisfiable.
That is, we translate the \( \deltazsar\)-formula \( R_{M_2(x)}(\ins{i}{Y}, Y, X, x) \) preserving its satisfiability.
The following CHCs are obtained by translating \( R_{M_2(x)}(\ins{i}{Y}, Y, X, x) \).
\begin{align}
    \sttoprd{q_0}(v_0, v_1, v_2, x) &\Leftarrow v_0 = i  \label{eq:intro:initial} \\
    \sttoprd{q_0}(v_0, v_1, v_2, x) &\Leftarrow \label{eq:intro:q0-q0} \sttoprd{q_0}(u_0, u_1, u_2, x)\wedge u_0 = u_1 + 1 \wedge v_0 = u_1 \wedge \lnot \ed \\
    \sttoprd{q_1}(v_0, v_1, v_2, x) &\Leftarrow \sttoprd{q_0}(u_0, u_1, u_2, x) \wedge u_0 = 0 \wedge u_2 = x \wedge v_0 = u_1 \wedge \lnot \ed \label{eq:intro:q0-q1} \\
    \sttoprd{q_1}(v_0, v_1, v_2, x) &\Leftarrow \sttoprd{q_1}(u_0, u_1, u_2, x) \wedge v_0 = u_1 \wedge \lnot \ed \label{eq:intro:q1-q1} \\
    \bot &\Leftarrow \sttoprd{q_1}(u_0, u_1, u_2, x) \wedge \ed \label{eq:intro:final}
\end{align}
Here \( \ed \defeq \bigwedge_{j \in \{0, 1, 2\}} \ispad(u_j) \).
The predicate \( \sttoprd{q} \) corresponds to the state \( q \) and \( \sttoprd{q}(v_0, v_1, v_2, x)\) intuitively means 
``there exists an assignment \(\alpha\) for \(X,Y\) and \(x\) such that
  given \((\ins{i}{\alpha(Y)},\alpha(Y),\alpha(X))\) as input,
  \(M_2(\alpha(x))\) visits state \(q\), with the next input letters being \((v_0,v_1,v_2)\)''.
The clause \eqref{eq:intro:initial} captures the fact that
\(M_2(\alpha(x))\) is initially at state \(q_0\), with the first element \(v_0\) of the initial input
is \( i \).
The clauses \eqref{eq:intro:q0-q0}, \eqref{eq:intro:q0-q1},
and \eqref{eq:intro:q1-q1} correspond to transitions
\( q_0 \trans{l_0 = l_1 + 1} q_0\), \( q_0 \trans{l_1 = 0 \wedge l_2 = x} q_1\),
and \(q_1\trans{\top}q_1\) respectively.
The constraints in the bodies of those clauses consist of:
(i) the labels of the transitions (e.g., \(u_0=u_1+1\) in \(\eqref{eq:intro:q0-q0}\)),
(ii) the equation \( u_1 = v_0\),
which captures the co-relation
between the arguments
\( \ins{i}{Y} \) and \( Y \) of \( R_{M_2}(x)\),
(iii) \( \lnot \ed \) indicating that there is still an input to read.
The last clause \eqref{eq:intro:final} captures the acceptance condition:
a contradiction is derived if \(M_2(x)\) reaches the final state \(q_1\), having read all the inputs.
It follows from the intuitions above that
the set of CHCs above is unsatisfiable, if and only if,
\( R_{M_2(x)}(\ins{i}{Y}, Y, X, x) \) is satisfiable.
}
\subsection{Translation}
\label{sec:translation}
We now formalize the translation briefly discussed in the previous subsection.
To simplify the definition of the translation, we first define the translation for terms in a special form.
Then we will show that every term of the form \( R_{M}(T_1, \ldots, T_k, t_1, \ldots, t_n)\) can be translated into the special form, preserving the satisfiability (or the unsatisfiability).

\begin{definition}
\newcommand{\listset}{\mathcal{L}}
Let \( \listset \) be a set of list terms.
Then \( \listset \) is
\begin{itemize}
  \item
    \emph{cons-free} if for all \( T \in \listset\), \( T \) is of the form \( \nil \) or \( \tailn{n}{X} \).
  \item
    \emph{gap-free} if \( \tailn{n}{X} \in \listset \) implies \( \tailn{m}{X} \in \listset \) for all \( 0 \le m \le n\).
\end{itemize}
Here \( \tailn{m}{X} \) is defined by \( \tailn{0}{X} \defeq X \) and \( \tailn{m + 1}{X} \defeq \tail{\tailn{m}{X}} \).

We say that a formula of the form \( R_{M(\seq{x})}(T_1, \ldots, T_k, t_1, \ldots, t_n)\) is \emph{normal} if \( \{ T_1, \ldots, T_k \} \) is cons-free and gap-free, and every \( t_i \) is an integer variable.
\end{definition}

\begin{definition}
Let \( R_{M(\seq{x})} \) be a \( (k, n) \)-ary SAR, where \( M(\seq{x}) = (Q, I, F, \Delta) \), \(\seq{T} \defeq T_0, \ldots, T_{k - 1}\) be list terms and \( \seq{y} \defeq y_0, \ldots, y_{n - 1} \) be integer variables.
Suppose that \( R_{M(\seq{x})}(\seq{T}, \seq{y}) \) is normal.
Then the set of CHCs translated from \(R_{M(\seq{x})}(\seq{T}, \seq{y}) \), written \( \toCHC{R_{M(\seq{x})}(\seq{T}, \seq{y})} \),
consists of the following clauses:
\begin{enumerate}
\item
  The clause (written \( \toCHC{p \trans{\phi} q} \)):
  \begin{align*}
    \sttoprd{q}(v_0, \ldots, v_{k - 1}, \seq{x}) \Leftarrow
    \begin{gathered}
      \sttoprd{p}(u_0, \ldots, u_{k - 1}, \seq{x}) \wedge \phi[u_0/l_0, \ldots, u_{k - 1}/ l_{k - 1}] \\
      \wedge \shift \wedge \padding \wedge \lnot \ed
    \end{gathered}
  \end{align*}
  for each \( p \trans{\phi} q \in \Delta\).
\item The clause
  \begin{align*}
    \sttoprd{q}(v_0, \ldots, v_{k - 1}, \seq{x}) \Leftarrow \cnstnil \wedge \tailnil \wedge \seq{x} = \seq{y}
  \end{align*}
  for each \(q\in I\).
\item The clause:
  \begin{align*}
    \bot \Leftarrow \sttoprd{q}(u_0, \ldots, u_{k - 1}, \seq{x}) \wedge \ed
  \end{align*}
  for each \(q\in F\).
\end{enumerate}
Here the definitions and the informal meanings of \(\ed \), \( \shift \), \(\padding \), \( \cnstnil \) and \( \tailnil \) are given as follows:
\allowdisplaybreaks[3]
\begin{align*}
  \ed &\defeq \bigwedge_{i \in \{0, \ldots, k - 1 \}}  \ispad(u_i) \quad\qquad \text{``there is no letter to read''}  \\
  \shift &\defeq \bigwedge \{ v_i = u_j \vee (\ispad(v_i) \wedge \ispad(u_j)) \mid T_i = \tailn{m}{X}, T_j = \tailn{m+1}{X}\} \\
  & \begin{aligned}
      &\text{``the \( l\)-th  element of the list represented by \( \tailn{m+1}{X}\) is}  \\
      &\text{ the \( (l + 1)\)-th element of the list represented by \( \tailn{m}{X} \)''}  \\
  \end{aligned} \\
  \padding &\defeq \bigwedge_{i \in \{0, \ldots, k-1 \}} \ispad{(u_i)} \Rightarrow \ispad{(v_i)} \\
  &\text{``if a padding symbol \( \pad \) is read, then the next input is also \( \pad \)''} \\
  \cnstnil &\defeq \bigwedge \{ \ispad(v_i) \mid T_i = \nil \} \\
  &\text{``there is no letter to read from an empty list''} \\
  \tailnil &\defeq \bigwedge \{\ispad(v_i) \Rightarrow \ispad(v_j) \mid T_i = \tailn{m}{X},  T_j = \tailn{m + 1}{X} \} \\
  & \begin{aligned}
     &\text{``if there is no letter to read from the input represented by \( \tailn{m}{X} \)}  \\
     &\text{then there is nothing to read from the input represented by \( \tailn{m+1}{X} \)''} \!\\
  \end{aligned}
\end{align*}
\end{definition}

We next show that we can assume that \( R_{M(\seq{x})}(\seq{T}, \seq{t}) \) is normal without loss of generality.
First, observe that ensuring that \( \seq{T} \) is gap-free is easy.
If \( \seq{T} \) is not gap-free then we just have to add additional inputs, corresponding to the list represented by \( \tailn{n}{X} \), to the automaton \( M(\seq{x}) \) and ignore those inputs.
Ensuring that \( \seq{t} \) is a sequence of integer variables is also easy.
If \( t \) is not an integer variable, then we can embed \( t \) to the transitions of the automaton and add the free variables of \( t \) to the parameter or as inputs of the automaton. %
Therefore, the only nontrivial condition is the cons-freeness:
\begin{lemma}
\label{lem:cons-free}
  Let \( R_{M(\seq{x})} \) be a \( (k, n) \)-ary SAR, \( \seq{T} \defeq T_1, \ldots, T_k \) be list terms and \( \seq{t} \defeq t_1, \ldots, t_n \) be integer terms.
  Then we can effectively construct a \( (k, n + m) \)-ary ss-NFA \( M'(\seq{x}, \seq{x}') \), list terms \( \seq{T'} \defeq T'_1, \ldots, T'_k \) and integer terms \( \seq{t'} \defeq t'_1, \ldots, t'_{n + m} \)such that (1) \(R_{M(\seq{x})}(\seq{T}, \seq{t})\) is satisfiable in \( \model \) iff \( R_{M'(\seq{x}, \seq{x}')}(\seq{T'}, \seq{t'}) \) is satisfiable in \( \model \) and (2) \( \{ T'_1, \ldots, T'_k \} \) is cons-free.
\qed
\end{lemma}
Instead of giving a proof, we look at an example.
(The proof is in Appendix~\ref{appx:cons-free}.)
Consider the formula \( \phi \defeq R_{M_2(x)}(\ins{1}{\ins{t}{Y}}, \ins{0}{Y}, X, x) \), where \( M_2(x) \) is the automaton given in Example~\ref{ex:automata-for-sorted-and-nth}.
We explain how to remove \( \fsymb{cons}(t, \cdot) \) from the first argument; by repeating this argument we can remove all the ``cons''.
Let \( \phi' \defeq R_{M'_2(x, y)}(\ins{1}{Y},\allowbreak \ins{0}{\tail{Y}}, X, x, t)\), where \( M_2'(x, y) \) is the ss-NFA in Fig.~\ref{fig:example-cons-free}.
Then it is easy to see that \( \phi \) is satisfiable in \( \model \) iff \( \phi' \) is satisfiable in \( \model \).
If \( \model, \alpha \models \phi \) with \( \alpha(Y) = w \) and \( \sem{t}_{\model, \alpha} = i\) then \( \model, \alpha[Y \mapsto iw] \models \phi' \); the opposite direction can be checked in a similar manner.
The idea is to embed the information that ``the second element is \( t \)'' into the ss-NFA by replacing \( l_0 \) with \( y \) in the edges that corresponds to the ``second step'' and passing \( t \) as the actual argument.
Note that we had to ``unroll the ss-NFA \( M_2(x) \)'' to ensure that the transition that contains \( y \) is used at most once.

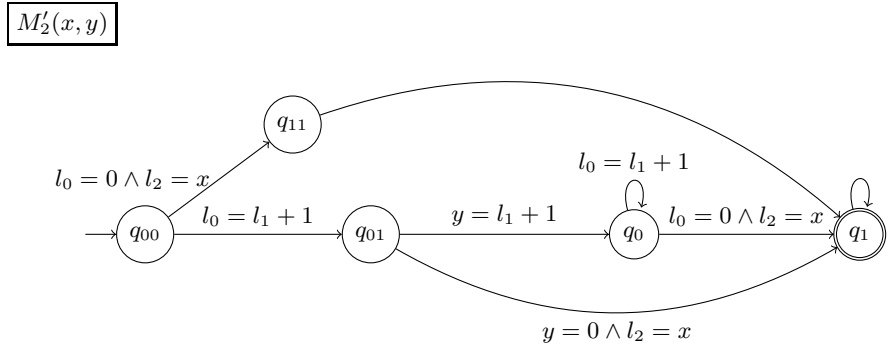
\begin{figure}[t]
\fbox{\( M'_2(x, y) \)}
\par
\centering
\begin{tikzpicture}
\node[state, initial] (q00) {\( q_{00} \)};
\node[state, right of=q00, xshift=1cm] (q01) {\( q_{01}\)};
\node[state, above right=of q00, yshift=-0.5cm] (q11) {\( q_{11} \)};
\node[state, right of=q01, xshift=1.5cm] (q0) {\( q_{0} \)};
\node[state, accepting,right of=q0, xshift=1cm] (q1) {\( q_1 \)};
\draw (q00) edge[above] node{\(l_0 = l_1 + 1 \)}  (q01)
      (q00) edge[left] node{\( l_0 = 0 \wedge l_2 = x\)}(q11)
      (q01) edge[above] node{\( y = l_1 + 1 \)} (q0)
      (q01) edge[bend right, below] node{ \( y = 0 \wedge l_2 = x\)} (q1)
      (q11) edge[bend left, above] node{}(q1)
      (q0) edge[loop above] node{\( l_0 = l_1 + 1 \)} (q0)
      (q0) edge[above] node{\(l_0 = 0 \land  l_2 = x \)} (q1)
      (q1) edge[loop above] (q1);
\end{tikzpicture}
\caption{The ss-NFA used to explain the idea behind Lemma~\ref{lem:cons-free}.}
\label{fig:example-cons-free}
\end{figure}

\begin{remark}
\label{rem:trees}
We have so far considered symbolic automatic relations on lists.
We expect that the reduction above can be extended to deal with
symbolic automatic relations on tree structures, as follows.
Let us consider a (symbolic, synchronous) bottom-up tree automaton,
with transitions of the form
\[\langle l_1,\ldots,l_k,x_1,\ldots,x_n\rangle(q_1,\ldots,q_m), \varphi \to q,\]
which means ``when the current node is labeled with
\(\langle l_1,\ldots,l_k,x_1,\ldots,x_n\rangle\) and the \(i\)-th child has been
visited with state \(q_i\), then the current node is visited with state \(q\)
if \(l_1,\ldots,l_k,x_1,\ldots,x_n\) satisfy \(\varphi\)''.
To reduce the satisfiability problem to CHCs on integers, it suffices to prepare
a predicate \(\sttoprd{q}\) for each state \(q\), so that
\(\sttoprd{q}(l_1,\ldots,l_k,x_1,\ldots,x_n)\) holds just if there exists an input that
allows the automaton to visit state \(q\) after reading
\(\langle l_1,\ldots,l_k,x_1,\ldots,x_n\rangle\).
As for the definition of ``normal form'', it suffices to replace \(\tail{T}\) with
\(\mathrm{child}_i(T)\) (which denotes the \(i\)-th child of tree \(T\)),
and define the cons-freeness and gap-freeness conditions accordingly.
The formalization of the extension is left for future work.
\end{remark}

\subsection{Correctness}
The correctness of the translation is proved by associating a derivation of \( \bot \) to an accepting run and vice versa.

We first define the notion of \emph{derivations} for CHCs, as a special case of the SLD resolution derivation~\cite{Kowalski74}.
Since the system of CHCs obtained by translating a \( \deltazsar \)-formula is \emph{linear}, which means that each clause contains at most one predicate in the body, we specialize the notion of derivations for linear CHCs.
\newcommand{\derive}[1]{\stackrel{#1}{\rightsquigarrow}}
\newcommand{\cnst}{\psi}
\newcommand{\dstate}[2]{\langle #1 \mid #2 \rangle}
\begin{definition}[Derivation]
A \emph{derivation state} (or simply a \emph{state}) is a pair \( \dstate{A}{ \cnst} \), where \( A \) is either \( \top \) or \( P(\seq{t})\), i.e.~an uninterpreted predicate symbol \( P \) applied to terms \( \seq{t}\), and \( \cnst \) is a constraint.
Let \( C \) be a linear constrained horn clause of the form \(P(\seq{t}_1) \Leftarrow A \wedge \cnst' \), where \( A \) is either \( \top \) or formula of the form \( Q(\seq{t}_2)\) and \( \cnst' \) is a constraint.
Then we write \( \dstate{P(\seq{t})}{\cnst} \derive{(C, \theta)} \dstate{\theta A}{\theta (\cnst \wedge \cnst') } \), if \( P(\seq{t}) \) and \( P(\seq{t}_1) \) are unifiable by a unifier \( \theta \).

Let \( \Pi \) be a system of linear CHCs, i.e.~a finite set of linear constrained horn clauses.
A \emph{derivation} from state \( S_0 \) with respect to \( \Pi \) is a finite sequence of the form \( S_0 \derive{(C_1, \theta_1)} S_1 \derive{(C_2, \theta_2)} \cdots \derive{(C_n, \theta_n)} S_n \) such that (i) \( C_i \in \Pi \) for all \( i \) and (ii) \( S_n \) is of the form \( \dstate{\top}{\cnst} \) such that \( \cnst \) is a constraint that is satisfiable in \( \Zp \).
\end{definition}

Now we are ready to prove the correctness of the translation. %

\begin{theorem}
\label{thm:correctness}
Let \( R_{M(\seq{x})} \) be a \( (k, n) \)-ary SAR, and suppose that \( R_{M(\seq{x})}(\seq{T}, \seq{y}) \) is normal.
Then \( R_{M(\seq{x})}(\seq{T}, \seq{y}) \) is satisfiable in \( \mathcal{M} \) iff \( \toCHC{R_{M(\seq{x})}(\seq{T}, \seq{y})} \) is unsatisfiable modulo \( \Zp \).
\end{theorem}
\begin{proof}[Sketch]
\newcommand{\formula}{R_{M(\seq{x})}(\seq{T}, \seq{y})}
\newcommand{\toterm}[1]{\underline{#1}}
We only sketch the proof; a detailed version is in Appendix~\ref{appx:correctness}.
Suppose that \( M(\seq{x}) = (Q, I, F, \Delta)\), \( \seq{T} = T_0, \ldots, T_{k - 1} \), \( \seq{y} = y_0, \ldots, y_{n - 1}\) and let \( \Pi \defeq \toCHC{\formula} \).
By the completeness of the SLD resolution, it suffices to show that \( \formula \) is satisfiable if and only if there is a derivation starting from \( \dstate{\sttoprd{q}(\seq{u}, \seq{x})}{\ed} \) with respect to \( \Pi \) for some \( q \in F \).
We separately sketch the proof for each direction.

(Only if)
Since \( \formula \) is satisfiable, there exists an assignment \( \alpha \) such that \( \model, \alpha \models \formula \).
Let \( w_i \defeq \sem{T_i}_{\model, \alpha}\) for each \( i \in \{0, \ldots, k - 1\}\), \( j_i \defeq \alpha(y_i)\) for \( i \in \{0, \ldots, n - 1 \} \) and \( \seq{j} \defeq j_0, \ldots, j_{n - 1} \).
Because \( \model, \alpha \models \formula \), we have an accepting run of \( M(\seq{j}) \),
\(q_0 \trans{a_0} q_1 \trans{a_2} \cdots \trans{a_{m - 1}} q_m \), where the run \( a_0a_2 \cdots  a_{m - 1} \) is \( \conv(w_0, \ldots, w_{k - 1}) \).
From this run, we can construct a derivation
\begin{align*}
  \dstate{\sttoprd{q_m}(\seq{u}, \seq{x})}{\!\ed}\!\! \derive{\xi_{m-1}} \!\! \dstate{\sttoprd{q_{m-1}}(a_{m - 1}, \seq{j})}{\cnst_{m - 1}} \! \derive{\xi_{m-2}} \!\!\! \cdots \! \derive{\xi_0} \! \dstate{\sttoprd{q_0}(a_0, \seq{j})}{\cnst_0} \! \derive{(C, \theta)} \! \dstate{\top}{\cnst}
\end{align*}
where \( \xi_i = (C_i, \theta_i)\).
Here, \( \sttoprd{q_{i}}(a_{i}, \seq{j}) \) means the predicate symbol \( \sttoprd{q_{i}} \) applied to constants that represent the elements of \( a_i \) and \( \seq{j} \).
In particular, the derivation can be constructed by taking the clause that corresponds to the transition \( q_{i} \trans{a_i} q_{i + 1}\) for the clause \( C_i \).

(If)
By assumption, there is a derivation
\begin{gather*}
\dstate{\sttoprd{q_m}(\seq{u}, \seq{x})}{\cnst_m}
\derive{(C_m, \theta_m)} \cdots
\derive{(C_1, \theta_1)}
\dstate{\sttoprd{q_0}(\seq{u}, \seq{x})}{\cnst_0} \derive{(C_0, \theta_0) } \dstate{\top}{\cnst}
\end{gather*}
where \( \cnst_m = \ed \).
We construct an accepting run of \( M(\seq{x}) \) using an assignment in \( \Zp \) and the unifiers.
Take an assignment \( \alpha \) such that \(\Zp, \alpha \models \cnst \), which exists because \( \cnst \) is satisfiable in \( \Zp \).
Let \( \theta_{\le i} \defeq \theta_0 \circ \theta_1 \circ \cdots \circ \theta_i \).
We define \( a_{ij} \in \Zp \), where \( 0 \le i \le m \) and \(0 \le j \le k - 1 \), by \( a_{ij} \defeq \sem{\theta_{\le i}(u_j)}_{\Zp, \alpha} \) and set \( a_i \defeq (a_{i0}, \ldots, a_{i k -1})\).
We also define \( j_i \defeq \sem{\theta_0(x_i)}_{\model, \alpha}\) and write \( \seq{j} \) for \( j_0, \ldots, j_n\).
Then we can show that \( a_0 a_1 \cdots a_{m - 1}\) is an accepting run of \( M(\seq{j}) \).
Moreover, we can show that \( a_0 a_1 \cdots a_{m - 1}\) can be given as a convolution of words \( \conv(w_0, \ldots, w_{k - 1}) \) by using the constraint \( \lnot \ed \) that appears in clauses corresponding to transition relations.
Finally, we can show that there is an assignment \( \beta \) in \( \model \) such that \( \sem{T_i}_{\model, \beta} = w_i \) for every \( i \in \{0, \ldots, k - 1\} \) by using the cons-freeness and gap-freeness, and the constraints \( \shift \), \( \cnstnil \) and \( \tailnil \).
\qed
\end{proof}

The correspondence between resolution proofs and accepting runs should allow us to generate a witness of the satisfiability of \( R_{M(\seq{x})}(\seq{T}, \seq{y}) \).
A witness of the satisfiability is important because it serves as a counterexample that the teacher part of an ICE-learning-based CHC solver provides
(cf.~\ref{sec:application-to-CHC}).
Since some CHC solvers like \eldarica\cite{HojjatRummer18} outputs a resolution proof as a certificate of
the unsatisfiability, it should be able to generate counterexamples by using these solvers as a backend of the teacher part.
The formalization and the implementation of this counterexample generation process are left for future work.

%% file: experiments.tex
\section{Experiments}
\label{sec:experiments}

We have implemented a satisfiability checker for
\( \deltaosar \)-formulas.
An input of the tool consists of (i) definitions of \(\deltaosar\)-predicates 
(expressed using ss-NFA), and
(ii) a \(\deltaosar\)-formula consisting of the defined predicates,
list constructors and destructors, and integer arithmetic. For (i),
if a predicate is defined using existential quantifiers, 
both the definitions of a predicate and its negation should be provided
(recall \(\mathit{nth}\) in Example~\ref{ex:automata-for-sorted-and-nth}); \skchanged{we do not need to provide the predicates in normal forms because our tool automatically translates the inputs into normal forms.}
\skchanged{The current version of our tool only outputs SAT, UNSAT, or TIMEOUT and does not provide any witness of the satisfiability.}
We used Spacer\cite{KomuravelliGC16}, \hoice\cite{ChampionCKS20}, and \eldarica\cite{HojjatRummer18} as the backend CHC solver (to solve the CHC problems obtained by the reduction
in Section~\ref{sec:reduction}).
The experiments were conducted on a machine with AMD Ryzen 9 5900X 3.7 GHz and 32 GB of memory, with a timeout of 60 seconds.
The implementation and all the benchmark programs are available in the artifact\cite{ShimodaKSS21}.
The detailed experimental results are shown in Appendix~\ref{sec:exp-detail}.

We have tested our tool for three benchmarks.
\skchanged{All the ss-NFAs used in the benchmarks are effective; in fact, all the formulas appearing as the labels of transitions are formulas in quantifier-free linear integer arithmetic.}
The first benchmark ``IsaPlanner'' is obtained from the benchmark\cite{JohanssonDB10} of IsaPlanner\cite{DixonFleuriot03}, which is a proof planner for the interactive theorem prover Isabelle\cite{Paulson94}.
We manually converted the recursively defined functions used in the original benchmark into SARs. %
The second benchmark ``SAR\_SMT'' consists of valid/invalid formulas that represent properties of lists.

Each instance of the third benchmark ``CHC'' consists of (i) CHCs on data structures
and (ii) a candidate model
(given as a map from predicate variables to \( \deltaosar \)-formulas);
the goal is to check that the candidate is a valid model for the CHCs (which
is the task of the ``teacher'' part of ICE-based CHC solving~\cite{ChampionCKS20}).
This benchmark includes
CHCs obtained by a reduction
from the refinement type-checking problem for functional programs~\cite{HashimotoUnno15,UnnoKobayashi09,UnnoTS17,ZhuJagannathan13}.
For many of the instances in the benchmark set, symbolic automatic relations
are required to express models.
For example, the set of CHCs given in Example~\ref{ex:chc-for-sorted} is included in the benchmark.

To compare our tool with the state-of-the-art SMT solvers, Z3 (4.8.11)~\cite{MouraBjorner08} and CVC4 (1.8)~\cite{BarrettCDHJKRT11},
which support user-defined data types and recursive function definition,
we manually translated the instances to
SMT problems that use recursive functions on lists.
We tested two different translations. %
One is to translate the \( \deltaosar \)-predicates (such as \( \pred{nth} \))
directly into recursive functions by using \texttt{define-fun-rec},
and the other is to translate the predicates %
into assertions, like \texttt{(assert (forall ...))}, that describe the definition of functions.

\begin{table}[tb]
\centering
{
  \def\mc#1#2#3{\multicolumn{#1}{#2}{#3}}
  \caption{Summary of the experimental results}
  \label{tbl:exp}
  \begin{tabular}{lrrrrrrrrrrr}
    \toprule
    Benchmark                     & \quad\textbf{IsaPlanner} & \quad\textbf{SAR\_SMT} & \quad\textbf{CHC} & \quad\textbf{All} \\ \midrule
    \#Instances                   & 15 (15/0)                & 60 (47/13)             & \quad 12 (12/0)   & \quad  87 (74/13) \\
    \textbf{Ours-Spacer} \\
      \quad \#Solved              & 8 (8/0/0)                & \quad 43 (30/13/0)           & \quad 8 (8/0/0)         & \quad 59 (46/13/0)      \\
      \quad Average time          & 0.995                    & 0.739                  & 1.981             & 0.942             \\
    \textbf{Ours-\hoice} \\
      \quad \#Solved              & 14 (14/0/1)              & 55 (42/13/0)           &\quad 11 (11/0/0)       & 80 (67/13/1)      \\
      \quad Average time          & 7.296                    & 4.498                  & 6.584             & 5.275             \\
    \textbf{Ours-Eldarica} \\
      \quad \#Solved              & 14 (14/0/0)              & 59 (46/13/2)           & 12 (12/0/0)       & 85 (72/13/2)      \\
      \quad Average time          & 4.539                    & 2.441                  & 11.078            & 4.006             \\
    \textbf{Z3 (rec)} \\
      \quad \#Solved              & 5 (5/0/0)                & 32 (19/13/0)           & 1 (1/0/0)         & 38 (25/13/0)      \\
      \quad Average time          & 0.023                    & 0.022                  & 0.017             & 0.022             \\
    \textbf{CVC4 (rec)} \\
      \quad \#Solved              & 5 (5/0/0)                & 32 (19/13/0)           & 3 (3/0/0)         & 40 (27/13/0)      \\
      \quad Average time          & 0.014                    & 0.015                  & 0.050             & 0.017             \\
    \textbf{Z3 (assert)} \\
      \quad \#Solved              & 7 (7/0/0)                & 20 (20/0/0)            & 3 (3/0/0)         & 30 (30/0/0)       \\
      \quad Average time          & 0.018                    & 0.018                  & 0.022             & 0.019             \\
    \textbf{CVC4 (assert)} \\
      \quad \#Solved              & 6 (6/0/0)                & 19 (19/0/0)            & 3 (3/0/0)         & 28 (28/0/0)       \\
      \quad Average time          & 0.057                    & 0.008                  & 0.015             & 0.019             \\
    \bottomrule
  \end{tabular}
}
\end{table}

Table~\ref{tbl:exp} summarizes the experimental results.
In the first column, ``Ours-XXX'' means our tool with the underlying CHC solver XXX,
``(rec)'' means the translation to recursive functions,
and ``(assert)'' means the translation to assertions.
The row “Benchmark” shows the names of the benchmarks.
The column ``All'' show the summary of the all benchmarks.
\changed{The row “\#Instances” shows the number of instances in the benchmark,
and the first two numbers in the parentheses show the numbers of valid and invalid instances respectively,
and the last number in the parentheses shows the number of solved instances that were not solved by the other tools.}
The row “\#Solved” shows the number of solved instances, and the numbers in the parentheses are the same as ones in ``\#Instances''.
The row ``Average time'' shows the average running time of the solved instances in seconds.

Ours-\eldarica successfully verified all the %
instances except two.
Since one of them needs non-linear properties on integers such as $x \geq y \times z$,
the reduced CHC problem cannot be proved by the underlying CHC solvers used in the experiments.
The other one is proved by Ours-\hoice.
As shown in the rows ``Z3'' and ``CVC4'',
many of the problems were not verified by the SMT solvers regardless of the way of translation.
Especially, they did not verify most of the instances that require inductions over lists.
Moreover, all the invalid instances translated by using assertions were not verified by Z3 nor CVC4,
while those
translated by using recursive functions were verified by Z3 and CVC4.

We explain some benchmark instances below.
The instance ``prop\_77'' in IsaPlanner benchmark is the correctness property of $\fsymb{insert}$ function of insertion sort.
That is, if a list $X$ is sorted and a list $Y$ is $X$ with some integer inserted by $\fsymb{insert}$,
then $Y$ is sorted.
As stated above, we manually converted the recursively defined functions into SARs.
As an example, we now describe how to translate $\fsymb{insert}$ function into a SAR.
The original $\fsymb{insert}$ function is defined as follows (written in OCaml-like language):
\begin{alltt}
let rec insert(x, y) = match y with
  | [] -> x :: []
  | z::xs -> if x <= z then x::y else z::insert(x, xs)
\end{alltt}
We first translate it into the following recursively defined predicate.
\begin{alltt}
let rec insert'(x, ys, rs) = match ys, rs with
  | [],     r::rs'             -> x = r && ys = rs'
  | y::ys', r::rs' when x <= y -> x = r && ys = rs'
  | y::ys', r::rs' when x > y  -> y = r && insert'(x, ys', rs')
  | _                          -> false
\end{alltt}
The predicate \verb|insert'(x,ys,rs)| means that \verb|insert(x,ys)| returns \verb|rs|.
We can now translate it into a SAR.
To express this predicate, we need two states---one for \verb|insert'| and one for equality of lists (\verb|ys = rs'|).
In addition, to check the equality of \verb|ys| and the tail of \verb|rs|,
we need a one-shifted list of \verb|ys| that has a dummy integer $0$ in its head, i.e., $\fsymb{cons}(0,\mathtt{ys})$.
Hence, predicate $\fsymb{insert}(x,X,Y)$ (which means $Y$ is $X$ with $x$ inserted)
can be expressed as $R_{M_{\texttt{ins}}(x)}(\fsymb{cons}(0,X), X, Y)$
where $M_{\texttt{ins}}(x)$ is shown in Fig.~\ref{fig:automaton-insert}.
The transition from $q_0$ to $q_0$ corresponds to the third case of the pattern matching of \verb|insert'|, and
the transition from $q_0$ to $q_1$ corresponds to the first two cases.

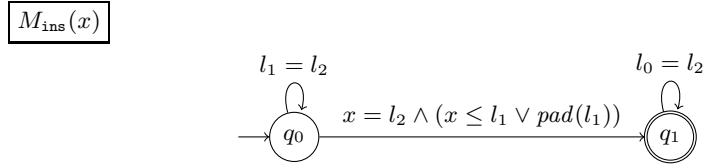
\begin{figure}[t]
\fbox{\( M_{\texttt{ins}}(x) \)}
\par
\centering
\begin{tikzpicture}
\node[state, initial] (q0) {\( q_0 \)};
\node[state, accepting,right of=q0, xshift=3cm] (q1) {\( q_1 \)};
\draw (q0) edge[loop above] node{\( l_1 = l_2 \)} (q0)
       (q0) edge[above] node{\(x = l_2 \land  (x \leq l_1 \lor \ispad(l_1)) \)} (q1)
       (q1) edge[loop above] node{\( l_0 = l_2 \)}  (q1);
\end{tikzpicture}
\caption{The ss-NFA for function $\fsymb{insert}$.}
\label{fig:automaton-insert}
\end{figure}

The instance ``prefix\_trans'' in SAR\_SMT benchmark is the transitivity property of predicate $\pred{prefix}$.
The predicate $\pred{prefix}$ takes two lists, and it holds if the first argument is the prefix of the second argument.
The transitivity of $\pred{prefix}$ is that, if $\pred{prefix}(X,Y)$ and $\pred{prefix}(Y,Z)$, then $\pred{prefix}(X,Z)$ holds.
The instance ``sorted'' in CHC benchmark is the problem explained in Example~\ref{ex:chc-for-sorted}.
All the instances explained here were solved by our tool, while neither Z3 nor CVC4 verified them.

%% file: related-work.tex
\section{Related Work}
\label{sec:related-work}

Although both automatic relations/structures~\cite{BlumensathGradel00,DBLP:conf/lcc/KhoussainovN94,DBLP:conf/lics/Gradel20} and symbolic automata~\cite{VeanesHT10,DBLP:conf/lpar/VeanesBM10,DBLP:conf/cav/DAntoniV17}
and their applications to verification have been well studied,
the combination of them is new to our knowledge, at least in the context of program verification.
D'Antoni and Veanes~\cite{DBLP:journals/fmsd/DAntoniV15} studied the notion of
extended symbolic finite automata (ESFA) which take a single word as an input,
but read multiple consecutive symbols simultaneously. ESFA is related
to our symbolic automatic relations in that
the language accepted by ESFA can be
expressed as \(\set{w\mid \REL(w, \tail{w},\ldots,\tailn{k-1}{w})}\) using a symbolic
automatic relation  \(\REL\).

Haudebourg~\cite[Chapter~6]{HaudebourgThesis} recently applied tree
automatic relations to CHC solving.  Since he uses ordinary
(i.e. non-symbolic) automata, his method can only deal with lists and
trees consisting of elements from a finite set.

As mentioned in Section~\ref{sec:intro}, the current SMT solvers do not work well for
recursive data structures. In the case of lists, one may use decidable theories on arrays or
inductive data
types~\cite{DBLP:conf/vmcai/BradleyMS06,DBLP:journals/jsat/BarrettST07}.
The decidable fragments of those theories are limited.
Our procedure is necessarily incomplete (due to the undecidability of the satisfiability problem), but
can be used as complementary to the procedures implemented in the current SMT solvers,
as confirmed by the experiments. We have focused on lists in this paper, but our approach can be
extended to deal with more general recursive data structures, by replacing automatic relations
with \emph{tree} automatic ones.

There are other approaches to solving CHCs on recursive data structures.
Unno et al.~\cite{UnnoTS17} proposed a method for automatically applying induction
on data structures, and De Angelis et al.~\cite{DBLP:journals/tplp/AngelisFPP18a,DeAngelisFPP20} proposed a method based on fold/unfold transformation.
An advantage of our approach is that we can generate a symbolic automatic relation as a certificate
of the satisfiability of CHCs. To make a proper comparison, however,
we have to devise and implement a missing component -- a procedure for automatically
generating a candidate model (recall that we have given only a procedure for checking the
validity of a candidate model).

%% file: conclusion.tex
\section{Conclusion}
\label{sec:conclusion}
We have introduced the notion of symbolic automatic relations (SARs) and
considered the satisfiability problem for SARs, with applications 
to SMT and CHC solving on recursive data structures in mind.
We have shown that the satisfiability problem 
is undecidable in general, but developed a sound (but incomplete)
procedure to solve the satisfiability problem by a reduction to CHC solving on integers.
We have confirmed the effectiveness of the proposed approach through experiments.
We plan to implement an ICE-based CHC solver based on the proposed approach.
To that end, we need to implement a learner's algorithm to automatically discover appropriate
SARs, following the approach of Haudebourg~\cite{HaudebourgThesis}.

%% file: reduction-appx.tex
\section{Supplementary Material for Section~\ref{sec:reduction}}

\subsection{Proof of Lemma~\ref{lem:cons-free}}
\label{appx:cons-free}
Here we give the proof of Lemma~\ref{lem:cons-free}, which has been omitted from the body of the paper.
For the idea behind the proof, please refer to Section~\ref{sec:translation}.
\begin{lemma}[Identical to Lemma~\ref{lem:cons-free}]
\label{lem:cons-free-restate}
  Let \( R_{M(\seq{x})} \) be a \( (k, n) \)-ary SAR, \( \seq{T} \defeq T_1, \ldots, T_k \) be list terms and \( \seq{t} \defeq t_1, \ldots, t_n \) be integer terms.
  Then we can effectively construct a \( (k, n + m) \)-ary ss-NFA \( M'(\seq{x}, \seq{x}') \), list terms \( \seq{T'} \defeq T'_1, \ldots, T'_k \) and integer terms \( \seq{t'} \defeq t'_1, \ldots, t'_{m + n} \) such that (1) \(R_{M(\seq{x})}(\seq{T}, \seq{t})\) is satisfiable in \( \model \) iff \( R_{M'(\seq{x}, \seq{x}')}(\seq{T'}, \seq{t'}) \) is satisfiable in \( \model \) and (2) \( \{ T'_1, \ldots, T'_k \} \) is cons-free.
\end{lemma}
\begin{proof}
First, note that if a term \( T \) is of the form \( \tail{\ins{t}{T'}} \) then the semantics of \( T\) is equal to that of \( T' \).
Thus, we may assume that each \( T_i \) is of the form \( \fsymb{cons}({t_{i 0}, \fsymb{cons}(t_{i 1}, \fsymb{cons}( \cdots \fsymb{cons}(t_{i l_i}, T''_i) \cdots )))}\), where \( T''_i \) is either \( \nil \) or \( \tailn{m}{X} \).
We fix a natural number \( l \defeq \max\{l_0, \ldots, l_{k - 1}\} \).
To simplify the notation, we write \( \ins{\seq{t}}{T} \) for the term \( \fsymb{cons}({t_{0}, \fsymb{cons}(t_{1}, \fsymb{cons}( \cdots \fsymb{cons}(t_{m}, T) \cdots )))}\) provided that \( \seq{t} = t_0, t_1, \ldots, t_{m}\).

\newcommand{\Mexpand}{M^{\mathrm{exp}}}
\newcommand{\Qexp}{Q^{\mathrm{exp}}}
\newcommand{\Iexp}{I^{\mathrm{exp}}}
\newcommand{\Fexp}{F^{\mathrm{exp}}}
\newcommand{\Deltaexp}{\Delta^{\mathrm{exp}}}

Before we construct the automaton \( M'(\seq{x}, \seq{x}')\) we first construct a ss-NFA \( \Mexpand(\seq{x}) \) that satisfies (i) given a word \( w \) such that \(|w| \ge l \), \( w \in \lang(\Mexpand(\seq{j})) \) iff \(w \in \lang(M(\seq{j})) \) and (ii) ``states that are reached in the first \( l \) steps are always distinct''.
Suppose that \( M(\seq{x}) = (Q, I, F, \Delta)\).
Such an automaton \( \Mexpand(\seq{x} ) \) can be defined as \( (\Qexp, \Iexp, \Fexp,\Delta^{\mathrm{exp}}) \) where
\begin{itemize}
    \item \( \Qexp \defeq Q \times \{0, \ldots, l\} \),
    \item \( \Iexp \defeq \{(q, 0) \mid q \in I \}\),
    \item \( \Fexp \defeq  \{ (q, l) \mid q \in F \} \) and
    \item \( \Deltaexp \defeq \{ ((p, i), \varphi, (q, i + 1)) \!\mid\! (p, \varphi ,q) \in \!\Delta, i \in \!\{ 0, \ldots, l - 1\}  \} \cup \{ ((p, l), \varphi, (q, l)) \!\mid (p, \varphi, q) \in \Delta \}\).
\end{itemize}
It is easy to check that this automaton \( \Mexpand(\seq{x}) \) satisfies the conditions (i) and (ii) described above.

Now we are ready to construct the ss-NFA \( M'(\seq{x}, \seq{x}') \).
We construct this automaton in a step-by-step manner and only show how each step work.

\newcommand{\Mstep}{M^{(i, m)}}
\newcommand{\Deltastep}{\Delta^{(i, m)}}
Suppose that \( T_i = \ins{\seq{s}}{\ins{s_m}{T''_i}} \), where \( \seq{s} = s_1, \ldots, s_{m - 1}\).
We describe how to remove the innermost \( \fsymb{cons}\) by constructing a new \( (k, n + 1)\)-ary ss-NFA \( \Mstep(\seq{x}, y) \) from \( \Mexpand(\seq{x})\) and integer terms \( \seq{t}'\) and list terms \( \seq{T}' \) such that
\begin{gather}
\label{subgoal-for-cons-free}
    \text{\( R_{M(\seq{x})}(\seq{T}, \seq{t})\) is satisfiable in \( \model \) iff \( R_{\Mstep(\seq{x}, y)}(\seq{T}', \seq{t}', s'_m) \) is satisfiable in \( \model \)}
\end{gather}
The set of states, initial states and final states of \( \Mstep(\seq{x}, y) \) are the same as that of \( \Mexpand(\seq{x}) \).
That is we define \( \Mstep(\seq{x}, y) \) as \( (\Qexp, \Iexp, \Fexp, \Deltastep) \) with an appropriate transition relation \( \Deltastep \) (as described below).
By repeating this process, we can ensure the cons-freeness.

We proceed by a case analysis on the shape of \( T_i'' \).
\begin{description}
    \item[(Case where \( T''_i = \nil \))]
      In this case, we set \( T'_i \defeq \ins{\seq{s}}{Y}\), where \( Y \) is a fresh variable, and \( T'_j \defeq T_j \) if \( j \neq i\).
      We also set \( \seq{t}' \defeq \seq{t}\) and \( s'_m \defeq s_m \).
      We define \( \Deltastep \) as
      \begin{align*}
        &\{ ((q, m - 1), \phi[y / l_i], (q', m) ) \mid ((q, m - 1), \phi, (q', m))) \in \Deltaexp \} \\
        &\cup \{ (  (q, i), \phi, (q', i')) \mid ((q, i), \phi, (q', i') \in \Deltaexp \text{ and } i < m - 1 \} \\
        &\cup \{ ( (q, i), \phi \wedge \ispad(l_i),  (q', i')) \mid ((q, i), \phi , (q', i') \in \Deltaexp \text{ and } m -1 < i  \}
      \end{align*}

       Now we check that the property \eqref{subgoal-for-cons-free} indeed holds.
       Suppose that \( \model, \alpha \models R_{M(\seq{x})}(\seq{T}, \seq{t}) \) and \( \sem{s_m}_{\model, \alpha} = j \).
       Then it is easy to check that \( \model, \alpha[Y \mapsto j ] \models R_{\Mstep(\seq{x}, y)}(\seq{T}', \seq{t}', s_m) \) because \( \sem{T'_i}_{\model, \alpha[Y \mapsto j ]} = \sem{T_i}_{\model, \alpha}\).
       Now suppose that \( \model, \alpha \models R_{\Mstep(\seq{x}, y)} (\seq{T}', \seq{t}', s'_m) \).
       Then \( \sem{Y}_{\model, \alpha} = j \) for some integer \( j \) (or \( \sem{Y} = \emptyword \)), otherwise we cannot obtain an accepting run of the automaton because of the condition \( \ispad(l_i) \) we added.
       So, \( \sem{T'_i}_{\model, \alpha} \) and \( \sem{T_i}_{\model, \alpha} \) must coincide except for the \( m \)-th element, which is \( \sem{s_m}_{\model, \alpha} \) for the list \( \sem{T_i}_{\model, \alpha} \) and \( j \) (or \( \pad \) if \( \sem{Y} = \emptyword \)) for \( \sem{T'_i}_{\model, \alpha} \).
       Since at the \( m \)-th step, the automaton \( \Mstep(\seq{x}, y)\) ignores the input from \( \sem{T'}_{\model, \alpha}\), and uses the integer \( \sem{s_m}_{\model, \alpha} \) instead, it follows that \( \model, \alpha \models R_{M(\seq{x})}(\seq{T}, \seq{t}) \).

    \item[(Case where \( T''_i = \tailn{r}{X} \))]
    For this case, we define \( \Deltastep \) as
    \begin{align*}
        &\{ ((q, m - 1), \phi[y / l_i], (q', m) ) \mid ((q, m - 1), \phi, (q', m))) \in \Deltaexp \} \\
        &\cup \{ ( (q, i), \phi, (q', i')) \mid ((q, i), \phi, (q', i') \in \Deltaexp \text{ and } i \neq m \}.
    \end{align*}
     We set \( T'_i \defeq \ins{\seq{s'}}{T''_i}\), where \( \seq{s'} = \seq{s}[\tail{X}/X]\), and \( T'_j \defeq T_j[\tail{X}/X] \) if \( j \neq i\).
     Similarly, we set \( t'_i \defeq t_i[\tail{X} / X]\) for each \( i \) and \( s_m' \defeq s_m [\tail{X} / X]\).

     We conclude the proof by checking the property \eqref{subgoal-for-cons-free}.
     Suppose that \(\model, \alpha \models R_{M(\seq{x})}(\seq{T}, \seq{t}) \) and \( \alpha(X) = w \).
     Then it is easy to check that \( \model , \alpha[X \mapsto 0w] \models R_{\Mstep(\seq{x}, y)} (\seq{T}', \seq{t}', s'_m) \).
     First, observe that
     \begin{align*}
      \sem{T'_i}_{\model, \alpha[X \mapsto 0w ]} 
      &= \sem{s'_1}_{\model, \alpha[X \mapsto 0w]} \cdots \sem{s'_{m - 1}}_{\model, \alpha[X \mapsto 0w]} \sem{T''_i}_{\model, \alpha[X \mapsto 0 w]} \\
      &= \sem{s_1}_{\model, \alpha} \cdots \sem{s_{m - 1}}_{\model, \alpha} \sem{T''_i}_{\model, \alpha[X \mapsto 0 w]} \\
      &= \sem{s_1}_{\model, \alpha} \cdots \sem{s_{m - 1}}_{\model, \alpha} \, a\,  \sem{T''_i}_{\model, \alpha}
     \end{align*}
     for some \( a \in \setZ \), which is the word obtained by replacing the \( m \)-th element of \( \sem{T_i}_{\model, \alpha} \) with \( a \).
     Since (1) \( a \) is ignored by the automaton \( \Mstep(\seq{x}, y)\), and \( \sem{s'_m}_{\model, \alpha[X \mapsto 0w]} = \sem{s_m}_{\model, \alpha} \) is used instead, (2) \( \sem{t_i}_{\model, \alpha} = \sem{t'_i}_{\model, \alpha[X \mapsto 0w]} \), and (3) \( \sem{T_j}_{\model, \alpha} = \sem{T'_j}_{\model, \alpha[X \mapsto 0w]} \) for \( j \neq i\), it follows that \( \model , \alpha[X \mapsto 0w] \models R_{\Mstep(\seq{x}, y)} (\seq{T}', \seq{t}', s'_m) \).

     Now assume that \( \model , \alpha \models R_{\Mstep(\seq{x}, y)} (\seq{T}', \seq{t}', s'_m) \).
     If \( \alpha(X) = i w \) for some \( i \in \setZ \) and \( w \in \setZ^* \), then we can show that \(\model, \alpha[X \mapsto w] \models R_{M(\seq{x})}(\seq{T}, \seq{t}) \).
     Otherwise, \( \alpha(X) = \emptyword \).
     Since we defined \( \pred{tail}^{\model}(\emptyword) = \emptyword \), we have \( \sem{T'_j}_{\model, \alpha} = \sem{T_j}_{\model, \alpha}\) for \( j \neq i \) and the same hold for \( \seq{t}' \), \( \seq{s}' \) and \( s'_m \).
     It follows that \(\model, \alpha \models R_{M(\seq{x})}(\seq{T}, \seq{t}) \).

\end{description}
\qed
\end{proof}

\subsection{Detailed proof of Theorem~\ref{thm:correctness}}
\label{appx:correctness}
We give a more detailed proof of Theorem~\ref{thm:correctness}.
In particular, we explain how to construct an assignment from a given derivation.

The following lemma is a  basic property of derivations that we will use to prove the correctness of the reduction.
We omit the proof for this lemma as it is a standard property of SLD-resolution.
\begin{lemma}
\label{lem:derivation-properties}
Let \( \Pi \) be a system of linear CHCs.
Suppose that \( \dstate{A_0}{\cnst_0} \derive{(C_1, \theta_1)} \dstate{A_1}{\cnst_1} \derive{(C_2, \theta_2)} \cdots \derive{(C_n, \theta_n) } \dstate{\top}{\cnst_n} \) is a derivation with respect to \( \Pi \) and let \( \alpha \) be an assignment such that \( \Zp, \alpha \models \cnst_n \).
  If \( \cnst_i \) is of the form \( \cnst_{i1} \wedge \cnst_{i2}\) (up to commutativity and associativity of conjunction) then \( \Zp, \alpha \models \theta_{i + 1}\theta_{i + 2} \cdots \theta_{n} \cnst_{i1}\).
  \label{it:assignment-sat-constraint}
\qed
\end{lemma}

\begin{proof}[Detailed proof of~Theorem~\ref{thm:correctness}]
\newcommand{\formula}{R_{M(\seq{x})}(\seq{T}, \seq{y})}
\newcommand{\toterm}[1]{\underline{#1}}
Suppose that \( M(\seq{x}) = (Q, I, F, \Delta)\), \( \seq{T} = T_0, \ldots, T_{k - 1} \), \( \seq{y} = y_0, \ldots, y_{n - 1}\) and let \( \Pi \defeq \toCHC{\formula} \).
As explained, it suffices to show that \( \formula \) is satisfiable if and only if there is a derivation starting from \( \dstate{\sttoprd{q}(\seq{u}, \seq{x})}{\ed} \) with respect to \( \Pi \) for some \( q \in F \).
We omit the only if direction.
This is because we think the proof of the only if case is not that difficult and the idea of the proof has already been given. %

(If)
By assumption, there is a derivation
\begin{gather*}
\dstate{\sttoprd{q_m}(\seq{u}_m, \seq{x}_m)}{\cnst_m}
\derive{(C_m, \theta_m)} \cdots
\derive{(C_1, \theta_1)}
\dstate{\sttoprd{q_0}(\seq{u}_0, \seq{x}_m)}{\cnst_0} \derive{(C_0, \theta_0) } \dstate{\top}{\cnst}
\end{gather*}
where \( \cnst_m = \ed \).
Here we assume that variables \( \seq{u}_m = u_{i0}, \ldots, u_{i(k - 1)}\) and \( \seq{x}_i = x_{i0}, \ldots, x_{i (n - 1)}\) are appropriately renamed to avoid name clash.
(We may simply write \( \seq{u} \) or \( u_{j} \) to mean \( \seq{u}_i \) or \( u_{ij}\) for some \( i \) that is clear from the context.)
Because a clause in \( \Pi \) that contains predicates in both the body and the head is a clause that corresponds to an element of the transition relation, for every \( C_i \) \( (1 \le i \le m) \), we have \( C_i = \toCHC{q \trans{\phi} p}\) for some transition \( (q \trans{\phi} p) \in \Delta \).

We construct an accepting run of \( M(\seq{x}) \) using an assignment in \( \Zp \) and the unifiers.
Take an assignment \( \alpha \) such that \(\Zp, \alpha \models \cnst \), which exists because \( \cnst \) is satisfiable in \( \Zp \).
Let \( \theta_{\le i} \defeq \theta_0 \circ \theta_1 \circ \cdots \circ \theta_i \).
We define \( a_{ij} \in \Zp \), where \( 0 \le i \le m \) and \(0 \le j \le k - 1 \), by \( a_{ij} \defeq \sem{\theta_{\le i}(u_{ij})}_{\Zp, \alpha} \) and set \( a_i \defeq (a_{i0}, \ldots, a_{i k -1})\).
Note that \( a_m = (\pad, \ldots, \pad) \) because \(\model, \alpha \models \theta_{\le m} \ed \) must hold by Lemma~\ref{lem:derivation-properties}.
We also define \( j_i \defeq \sem{\theta_0(x_i)}_{\model, \alpha}\) and write \( \seq{j} \) for \( j_0, \ldots, j_n\).
Our goal is to show that \( a_0 a_1 \cdots a_{m - 1}\) is an accepting run of \( M(\seq{j}) \) such that
\begin{gather*}
  \text{\( q_0 \trans{a_0} q_1 \trans{a_1} \cdots \trans{a_{m - 1}} q_m \) and \( q_{i - 1} \trans{a_{i-1}} q_{i} \) corresponds to \( C_i \).}
\end{gather*}
The fact that \( q_0 \) is initial and \( q_m \) is final is obvious because only predicate symbols corresponding to initial (resp.~final) states can appear in the initial (resp.~final) derivation state.
Thus, it suffices to check that \( q_i \trans{a_i} q_{i + 1} \) is a valid transition for every \( i \in \{ 0, \ldots, m - 1\} \).
First observe that \( \sem{\theta_{\le i}(x_l)} = j_l\) for any \( i \in \{0, \ldots, m \} \). %
Recall that, for each \( i \), we have \( \dstate{\sttoprd{q_i}(\seq{u}, \seq{x})}{\cnst_i} \derive{(\toCHC{q_{i - 1} \trans{\phi} q_i}, \theta_i)} \dstate{\sttoprd{q_{i-1}}(\seq{u},\seq{x})}{\cnst_{i - 1}}\) for some \( \phi \).
Thus, \( \cnst_{i - 1} \) is of the form \( \theta_i(\phi [u_0/l_0, \ldots, u_{k - 1} / l_{k - 1}]) \wedge \cnst' \) for some \( \cnst' \).
By Lemma~\ref{lem:derivation-properties}, \( \Zp, \alpha \models \theta_{\le i} (\phi [u_0/l_0, \ldots, u_{k - 1} / l_{k - 1}])\) must hold.
Since the free variables of \( \phi \) are among \( \seq{l}, \seq{x} \), we have \( \Zp, [\seq{l} \mapsto a_i, \seq{x} \mapsto \seq{j}] \models \phi\) as desired.

Next we show that the accepting run \( a_0 a_1 \cdots a_{m - 1}\) can be given as a convolution of words \( \conv(w_0, \ldots, w_{k - 1}) \); note that such \( w_0, \ldots, w_{k-1} \in \setZ^* \) are unique if they exist.
Let \( w'_i \defeq  a_{0i} a_{1i} \cdots a_{m - 1 i} \).
It suffices to show that (i) for each \( i \), \( w'_i = w_i w''_i \) for some \( w_i \in \setZ^* \) and \( w''_i \in \{ \pad \}^* \) and (ii) \( w_i' \in \setZ^* \) for some \( i \).
Then we have \( \conv(w_0, \ldots, w_{l - 1}) = a_0a_1 \cdots a_{m - 1} \).
To prove (i), first note that, for each \( i \in \{0, \ldots, m - 1 \}\),  \( \cnst_i \) must be of the form \( \theta_i \padding   \wedge \cnst'\) because \( \padding \) appears in the body of \( C_{i + 1} \).
Since \( \model, \alpha \models \theta_{\le i} \padding \) holds Lemma~\ref{lem:derivation-properties}, it follows that ``if \( a_{ij} = \pad \) then \( a_{{i + 1} j} = \pad \)''.
Hence, (i) holds.
To prove (ii) first observe that \( \cnst_{m - 1} \) must be of the form \(  \lnot \ed \wedge \cnst' \) with the free variables of \( \lnot \ed \) being \( \seq{u}_{m - 1}\).
Thus, by Lemma~\ref{lem:derivation-properties}, we have \( \model, \alpha \models \theta_{\le {m -1}} (\lnot \ed) \).
This implies that there exists  \( i \) such that \( a_{{m - 1} i} \in \setZ \).
Therefore, from (i), we have \( w'_i \in \setZ^* \).

It remains to show that there is an assignment \( \beta \) in \( \model \) such that \( \sem{T_i}_{\model, \beta} = w_i \) for every \( i \in \{0, \ldots, k - 1\} \).
Since \( \{ T_0, \ldots, T_{k - 1}\}\) is gap-free, each list variable \( X \) that is free in \( \seq{T} \) must be equal to some \( T_i \).
We define \( \beta \) by \( \beta(X) \defeq w_i \) provided that \( X = T_i \).
What we need to verify is  (i) \( w_i = \emptyword \) if \( T_i = \nil \) and (ii) \( w_i = \sem{\tailn{m}{X}}_{\model, \beta }\) if \( T_i = \tailn{m}{X} \).
It is easy to check that (i) holds because, by the constraint \( \cnstnil \), \( a_{0i} \) must be \( \pad \) if \( T_i = \nil \).
The statement (ii) is proved by induction on \( m \).
The case where \( m = 0 \) is trivial.
Suppose that \( 0 < m \), \( T_{j'} = \tailn{m - 1}{X} \) and \(T_j = \tailn{m}{X} \).
By the induction hypothesis, we have \( \sem{T_{j'}}_{\model, \beta} = w_{j'} \).
If \( w_{j'} = \emptyword \) then we can easily check that \( w_j = \emptyword \) using the constraint \( \tailnil \).
So, let us consider the case where \( w_j \neq \emptyword \).
Observe that, for each \( i \in \{0, \ldots, m - 1 \}\),  \( \cnst_i \) must be of the form \(\theta_i \shift \wedge \cnst'\) because \( \shift \) appears in the body of \( C_{i + 1} \).
By Lemma~\ref{lem:derivation-properties}, we have \(\model, \alpha \models  \theta_{\le i} \shift \).
We thus have \( a_{ij'} = a_{(i - 1)j} \) for all \( i \in \{1, \ldots, m \} \).
(Note that in particular we have \( a_{(m - 1) j} = \pad \) because \(a_{m j} = \pad \)).
Therefore, \( w_j \) is the tail of \( w_j' \) as desired.

\qed
\end{proof}

%% file: exp_detail.tex
\section{Details of the Experiments}
\label{sec:exp-detail}

{
\newcommand\cons{\mathit{cons}}
\newcommand\cnt{\mathit{count}}
\newcommand\mem{\mathit{count}}
\newcommand\length{\mathit{length}}
\newcommand\take{\mathit{take}}
\newcommand\sm{\mathit{sum}}
\newcommand\prefix{\mathit{prefix}}
\newcommand\insone{\mathit{ins1}}
\newcommand\inse{\mathit{ins}}
\newcommand\insort{\mathit{insort}}
\newcommand\butlast{\mathit{butlast}}
\newcommand\sorted{\mathit{sorted}}
\newcommand\last{\mathit{last}}
\newcommand\nth{\mathit{nth}}
\begin{table}[p]
  \caption{Lists of the instances in the benchmark sets (1/2)}
  \label{tbl:benchmark1}
  \centering
  \begin{tabular}{@{}lllr@{}}
    \toprule
    Benchmark & Instance & Property & Validity \\ \midrule
    IsaPlanner & prop\_04 & $1 + \cnt(n,X) = \cnt(n,cons(n,X))$ & \checkmark \\
    IsaPlanner & prop\_05 & $n = x \Rightarrow 1 + \cnt(n, X) = \cnt(n, \inse(x,X))$ & \checkmark \\
    IsaPlanner & prop\_15 & $\length (\inse( x, X)) = 1 + \length( X)$ & \checkmark \\
    IsaPlanner & prop\_16 & $X = \nil \Rightarrow \last (\inse(x,X)) = x$ & \checkmark \\
    IsaPlanner & prop\_29 & $\mem( x, (\insone(x, X)))$ & \checkmark \\
    IsaPlanner & prop\_30 & $\mem( x, (\inse( x, X)))$ & \checkmark \\
    IsaPlanner & prop\_39 & $\cnt(n, X) + \cnt(n, \cons(m,\nil)) = \cnt(n, cons(m,  X))$ & \checkmark \\
    IsaPlanner & prop\_40 & $\take(0, X) = \nil$ & \checkmark \\
    IsaPlanner & prop\_42 & $n \geq 0 \land Y = \take(n, X) \land zs = \take(n+1, \inse(x,X)) $ & \checkmark \\
               &          & $\qquad \Rightarrow zs = \inse(x,Y)$ &  \\
    IsaPlanner & prop\_50 & $Y = \butlast(X) \land n = \length( X) \Rightarrow Y = \take(n-1, X)$ & \checkmark \\
    IsaPlanner & prop\_62 & $X \neq \nil \Rightarrow \last (cons(x,X)) = \last(X)$ & \checkmark \\
    IsaPlanner & prop\_67 & $\length (\butlast(X)) = \length( X) - 1$ & \checkmark \\
    IsaPlanner & prop\_71 & $x \neq y \Rightarrow \mem( x, (\inse( y, X))) = \mem( x, X)$ & \checkmark \\
    IsaPlanner & prop\_77 & $\sorted( X) \land Y = \insort(x, X) \Rightarrow \sorted( Y)$ & \checkmark \\
    IsaPlanner & prop\_86 & $x < y \Rightarrow \mem( x, (\inse( y, X))) = \mem( x, X)$ & \checkmark \\
    SAR\_SMT & count\_cons & $\cnt(x, X) \leq \cnt(x, \inse(y,X))$ & \checkmark \\
    SAR\_SMT & count\_cons\_eq & $x = y \land n = \cnt(x, X) \Rightarrow n + 1 = \cnt(x, \inse(y,X))$ & \checkmark \\
    SAR\_SMT & count\_cons\_invalid & $\cnt(x, X) = \cnt(x, \inse( y, X))$ &  \\
    SAR\_SMT & count\_cons\_neq1 & $x \neq y \Rightarrow \cnt(x, X) = \cnt(x, \inse(y,X))$ & \checkmark \\
    SAR\_SMT & count\_cons\_neq2 & $x \neq y \Rightarrow \cnt(x, X) = \cnt(x, \inse(y,X))$ & \checkmark \\
    SAR\_SMT & count\_nil1 & $\cnt(x, \nil) = n \Rightarrow n = 0$ & \checkmark \\
    SAR\_SMT & count\_nil2 & $\cnt(x, \nil) = 0$ & \checkmark \\
    SAR\_SMT & ins\_head1 & $x < y \Rightarrow \inse( x, cons(y,X)) = \inse(x,cons(y,X))$ & \checkmark \\
    SAR\_SMT & ins\_head2 & $x \geq y \land \inse( x, X) = Y$ & \checkmark \\
             &            & $\qquad \Rightarrow \inse( x, cons(y,X)) = \cons(y , \cons(x,Y))$ & \\
    SAR\_SMT & ins\_insort & $\sorted( X) \Rightarrow \inse( x, X) = \insort(x, X)$ & \checkmark \\
    SAR\_SMT & ins\_insort\_invalid & $\inse( x, X) = \insort(x, X)$ &  \\
    SAR\_SMT & ins\_nil & $X = \nil \Rightarrow \inse(x, X) = \cons(x,\nil)$ & \checkmark \\
    SAR\_SMT & last\_nil & $X = \nil \Rightarrow \last(\nil) = 0$ & \checkmark \\
    SAR\_SMT & last\_singleton & $Y = \cons(x,\nil) \Rightarrow \last(Y) = x$ & \checkmark \\
    SAR\_SMT & length\_2 & $\length(\cons(x,\cons(y,\nil))) \geq 3$ &  \\
    SAR\_SMT & length\_cons1 & $n = \length (cons(x,X)) \land n' = \length( X) \Rightarrow n = n' + 1$ & \checkmark \\
    SAR\_SMT & length\_cons2 & $n = \length( X) \Rightarrow n + 1 = \length (cons(x,X))$ & \checkmark \\
    SAR\_SMT & length\_cons\_invalid & $\length (cons(x,X)) = \length( X)$ &  \\
    SAR\_SMT & length\_count & $\length( X) \geq \cnt(x, X)$ & \checkmark \\
    SAR\_SMT & length\_count\_invalid & $\length( X) > \cnt(x, X)$ &  \\
    SAR\_SMT & length\_nat & $\length( X) \geq 0$ & \checkmark \\
    SAR\_SMT & length\_nil & $X = \nil \Rightarrow \length( X) = 0$ & \checkmark \\
    SAR\_SMT & length\_non\_nil & $\length (cons(x,X)) \geq 1$ & \checkmark \\
    SAR\_SMT & prefix\_antisymetric & $\prefix( X, Y) \land \prefix( Y, X) \Rightarrow X = Y$ & \checkmark \\
    SAR\_SMT & prefix\_cons & $\prefix( X, cons(x,X)) \Rightarrow \forall y \in X.\  x = y$ & \checkmark \\
    SAR\_SMT & prefix\_cons\_invalid & $\prefix( X, cons(x,X))$ &  \\
    SAR\_SMT & prefix\_count & $\prefix( X, Y) \Rightarrow \cnt(x, X) \leq \cnt(x, Y)$ & \checkmark \\
    SAR\_SMT & prefix\_count\_invalid & $\prefix( X, Y) \Rightarrow \cnt(x, X) > \cnt(x, Y)$ &  \\
    \bottomrule
  \end{tabular}
\end{table}
\begin{table}
\caption{Lists of the instances in the benchmark sets (2/2)}
  \label{tbl:benchmark2}
  \centering
  \begin{tabular}{@{}lllr@{}}
    \toprule
    Benchmark & Instance & Property & Validity \\ \midrule
    SAR\_SMT & prefix\_exists & $\prefix( X, Y) \land 0 \in X \Rightarrow 0 \in Y$ & \checkmark \\
    SAR\_SMT & prefix\_exists\_invalid & $\prefix( X, Y) \land 0 \in Y \Rightarrow 0 \in X$ &  \\
    SAR\_SMT & prefix\_forall & $\prefix( X, Y) \land (\forall y\in Y.\  0 \leq y) \Rightarrow \forall x\in X.\  0 \leq x$ & \checkmark \\
    SAR\_SMT & prefix\_forall\_invalid & $\prefix( X, Y) \land (\forall x\in X.\  0 \leq x) \Rightarrow \forall y\in Y.\  0 \leq y$ &  \\
    SAR\_SMT & prefix\_hd & $\prefix( (cons(x,X)), (cons(y,Y))) \Rightarrow x = y$ & \checkmark \\
    SAR\_SMT & prefix\_length & $\prefix( X, Y) \Rightarrow \length( X) \leq \length( Y)$ & \checkmark \\
    SAR\_SMT & prefix\_non\_nil & $\prefix( (cons(x,X)), Y) \Rightarrow Y \neq \nil$ & \checkmark \\
    SAR\_SMT & prefix\_nth & $\prefix( X, Y) \land x = \nth(X,i) \Rightarrow x = \nth(Y,i) $ & \checkmark \\
    SAR\_SMT & prefix\_sum & $\prefix( X, Y) \land (\forall y\in Y.\  0 \leq y) \Rightarrow \sm(X) \leq \sm(Y)$ & \checkmark \\
    SAR\_SMT & prefix\_sum\_invalid & $\prefix( X, Y) \Rightarrow \sm(X) \leq \sm(Y)$ &  \\
    SAR\_SMT & prefix\_trans & $\prefix( X, Y) \land \prefix( Y, zs) \Rightarrow \prefix( X, zs)$ & \checkmark \\
    SAR\_SMT & scan\_sum\_length & $Y = \sm(X) \Rightarrow \length( X) = \length( Y)$ & \checkmark \\
    SAR\_SMT & scan\_sum\_pos & $Y = \sm(X) \land (\forall y\in Y.\  0 \leq y) \Rightarrow \forall x\in X.\  0 \leq x$ & \checkmark \\
    SAR\_SMT & scan\_sum\_pos\_invalid & $Y = \sm(X) \land (\forall x\in X.\  0 \leq x) \Rightarrow \forall y\in Y.\  0 \leq y$ &  \\
    SAR\_SMT & sorted & $sorted( X) \land head X \geq 0 \Rightarrow 0 \leq \nth(X, i)$ & \checkmark \\
    SAR\_SMT & sorted\_forall & $\sorted( (cons(x,X))) \Rightarrow \forall x\in X.\  0 \leq x$ & \checkmark \\
    SAR\_SMT & sorted\_neg & $\neg (\sorted( (cons(x,X))) \land (\exists x' \in X.\   0 > x') \land x \geq 0)$ & \checkmark \\
    SAR\_SMT & sorted\_nth & $\sorted( (cons(x,X))) \Rightarrow x \leq \nth(X,i)$ & \checkmark \\
    SAR\_SMT & sorted\_pos & $\sorted( (cons(x,X))) \land 0 \leq x \Rightarrow \forall x\in X.\  0 \leq x$ & \checkmark \\
    SAR\_SMT & sorted\_prefix & $\sorted( Y) \land \prefix( X, Y) \Rightarrow \sorted( X)$ & \checkmark \\
    SAR\_SMT & sorted\_singleton & $\sorted( (cons(x,X))) \land (\exists x' \in X.\   x > x') \Rightarrow X = \nil$ & \checkmark \\
    SAR\_SMT & sumfold\_gte & $(\forall x\in X.\  0 \leq x) \land \sm(X) \geq x \times \length(X)$ & \checkmark \\
    SAR\_SMT & take\_cons & $n > 0 \land Y = \take(n-1, X)$ & \checkmark \\
             &            & $\qquad \Rightarrow \cons(x,Y) = \take(n, (\cons(x,X)))$ &  \\
    SAR\_SMT & take\_exists & $Y = \take(n, X) \land 0 \in X \Rightarrow 0 \in Y$ & \checkmark \\
    SAR\_SMT & take\_length & $n \geq 0 \land Y = \take(n, X) \Rightarrow \length( Y) \Rightarrow n$ & \checkmark \\
    SAR\_SMT & take\_length\_invalid & $Y = \take(n, X) \Rightarrow \length( Y) = n$ &  \\
    SAR\_SMT & take\_nil & $\take(n, X) = \nil$ & \checkmark \\
    SAR\_SMT & take\_sum & $Y = \take'( n, X) \land (\forall x\in X.\  0 \leq x) \Rightarrow \sm(Y) \leq \sm(X)$ & \checkmark \\
    SAR\_SMT & take\_total\_count & $Y = \take'( n, X) \Rightarrow \cnt(x, X) \leq \cnt(x, Y)$ & \checkmark \\
    SAR\_SMT & take\_total\_length\_eq & $Y = \take'( n, X) \land 0 \leq n \leq \length( Y) \Rightarrow n = \length( X)$ & \checkmark \\
    SAR\_SMT & take\_total\_length\_invalid & $Y = \take'( n, X) \land n \leq \length( Y) \Rightarrow n = \length( X)$ &  \\
    SAR\_SMT & take\_total\_length\_leq & $X = \take'( n, Y) \land n_1 = \length( X) \land n_2 = \length( Y)$ & \checkmark \\
             &                          & \qquad$\Rightarrow n_1 \leq n_2$ &  \\
    CHC & compare\_length & properties on comparisons of list lengths & \checkmark \\
    CHC & fib & a property on the length of lists of the fibonacci numbers  & \checkmark \\
    CHC & filter & VC: \texttt{filter} reduces the length of the given list & \checkmark \\
    CHC & insert\_sort & VC: sortedness of insertion sort & \checkmark \\
    CHC & insertion & sortedness of insertion & \checkmark \\
    CHC & leq\_list & antisymetricity of the dictionary order on lists & \checkmark \\
    CHC & make\_length & VC: \texttt{make} generates a list of the given length  & \checkmark \\
    CHC & merge\_sort & VC: sortedness of merge sort & \checkmark \\
    CHC & rev\_length & VC: \texttt{reverse} preserves lengths & \checkmark \\
    CHC & same & $(\forall x\in X.\ x = y) \land (\exists x\in X.\ x = y+1) \Rightarrow \bot$ & \checkmark \\
    CHC & sorted & $\sorted(\cons(1,X)) \land (0 \in X) \Rightarrow \bot $ & \checkmark \\
    CHC & sum & a property on element-wise summation & \checkmark \\
    \bottomrule
  \end{tabular}
\end{table}
}

Tables~\ref{tbl:benchmark1} and \ref{tbl:benchmark2} show the lists of the instances of the benchmark sets used in the experiments.
The column ``Instance'' shows the name of the instances.
The column ``Property'' shows the properties represented in the instances.
The column ``Validity'' shows the validity of the property. The checkmark $\checkmark$ indicates that the property holds.
In the benchmark ``CHC'', ``VC:'' means the verification condition for refinement types.
We use a set-like notation $x \in X$ for that $x$ is an element of the list $X$.

Tables~\ref{tbl:exp-detail1} and \ref{tbl:exp-detail2} show the experimental results.
The numbers in the tables show the running time of the tool for the instance in seconds.
The empty cell means the tool does not verify the instance by the timeout or aborts by out of memory.
The only instance that is not verified by our tool is ``sumfold-gte'', which is not verfied by also the other solvers.

\begin{table}[tb]
\centering
{
  \def\mc#1#2#3{\multicolumn{#1}{#2}{#3}}
  \caption{Experimental results (1/2)}
  \label{tbl:exp-detail1}
  \begin{tabular}{lrrrrrrr}
    \toprule
    & \multicolumn{3}{c}{\textbf{Ours}} & \multicolumn{2}{c}{(rec)} & \multicolumn{2}{c}{(assert)} \\ \cmidrule(r){2-4} \cmidrule(r){5-6} \cmidrule(r){7-8}
    & \textbf{Spacer} & \quad\textbf{\hoice} & \quad\textbf{Eldarica} & \quad\textbf{Z3} & \quad\textbf{CVC4} & \quad\textbf{Z3} &\quad\textbf{CVC4} \\ \midrule
\textbf{IsaPlanner} \\
\quad prop\_04.smt2 & & 1.891&3.741 & & &\ 0.018&0.012\\
\quad prop\_05.smt2 & & 4.756&5.42 & &  & & \\
\quad prop\_15.smt2 & & 4.103&3.725 & &  & & 0.287\\
\quad prop\_16.smt2&0.823&0.877&1.531&\ 0.019&0.009&0.018&0.006\\
\quad prop\_29.smt2&0.17&0.168&0.657&0.019&0.012&0.019&0.013\\
\quad prop\_30.smt2&0.179&0.178&0.649&0.019&0.014 & & \\
\quad prop\_39.smt2 & & &6.151 & & &0.017&0.016\\
\quad prop\_40.smt2&0.285&0.3&1.003&0.017&0.006&0.018&\\
\quad prop\_42.smt2 & & 58.646&8.199 & & &0.013&0.007\\
\quad prop\_50.smt2 & & 0.922 & &  & & &\\
\quad prop\_62.smt2&2.837&3.266&3.917&0.041&0.0209&0.020&\\
\quad prop\_67.smt2 & & 16.086&16.73 & &  & & \\
\quad prop\_71.smt2&1.644&4.566&4.849 & &  & & \\
\quad prop\_77.smt2&0.456&1.067&2.227 & &  & & \\
\quad prop\_86.smt2&1.564&5.321&4.741 & &  & & \\
\textbf{SAR\_SMT} \\
\quad count\_cons.smt2 & & 13.543&3.536&0.018&0.011&0.018&0.008\\
\quad count\_cons\_eq.smt2 & & 1.913&2.857&0.014&0.012&0.021&0.01\\
\quad count\_cons\_invalid.smt2&1.275&1.613&3.462&0.021&0.015 & & \\
\quad count\_cons\_neq1.smt2 & & 20.32&5.371 & & &0.016&0.008\\
\quad count\_cons\_neq2.smt2 & & &3.55&0.017&0.014&0.021&0.008\\
\quad count\_nil1.smt2&0.434&0.432&0.916&0.019&0.009&0.018&0.007\\
\quad count\_nil2.smt2&0.092&0.087&0.58&0.019&0.010&0.019&0.008\\
\quad ins\_head1.smt2&0.608&0.613&1.174&0.019&0.007&0.017&0.009\\
\quad ins\_head2.smt2&1.993&5.792&4.222&0.016&0.012&0.021&0.009\\
\quad ins\_insort.smt2&0.784&4.061&2.073 & &  & & \\
\quad ins\_insort\_invalid.smt2&0.716&1.165&1.847&0.020&0.031 & & \\
\quad ins\_nil.smt2&0.246&0.261&0.926&0.020&0.013&0.020&0.005\\
\quad last\_nil.smt2&0.426&0.42&0.893&0.018&0.009&0.015&0.006\\
\quad last\_singleton.smt2&0.495&0.66&1.276&0.019&0.007&0.018&0.005\\
\quad length\_2.smt2&0.639&0.65&1.325&0.018&0.013 & & \\
\quad length\_cons1.smt2 & & 2.705&2.542&0.018&0.010&0.019&0.007\\
\quad length\_cons2.smt2 & & 0.528&1.446&0.020&0.010&0.019&0.01\\
\quad length\_cons\_invalid.smt2&1.209&1.421&2.44&0.017&0.01 & & \\
\quad length\_count.smt2 & & 29.491&2.628 & &  & & \\
\quad length\_count\_invalid.smt2&1.257&1.355&2.485&0.019&0.007 & & \\
\quad length\_nat.smt2&0.493&0.534&1.312 & &  & & \\
\quad length\_nil.smt2&0.155&0.114&0.743&0.036&0.028&0.018&0.008\\
\quad length\_non\_nil.smt2&0.5&0.528&1.364 & &  & & \\
\quad prefix\_antisymetric.smt2&0.148&0.087&0.678 & &  & & \\
\quad prefix\_cons.smt2&0.184&0.2&0.827 & &  & & \\
\quad prefix\_cons\_invalid.smt2&0.161&0.184&0.962&0.019&0.01 & & \\
\quad prefix\_count.smt2 & & &2.666 & &  & & \\
\quad prefix\_count\_invalid.smt2&1.31&1.794&2.705&0.017&0.011 & & \\
\bottomrule
  \end{tabular}
}
\end{table}
\begin{table}[tb]
\centering
{
  \def\mc#1#2#3{\multicolumn{#1}{#2}{#3}}
  \caption{Experimental results (2/2)}
  \label{tbl:exp-detail2}
  \begin{tabular}{lrrrrrrr}
    \toprule
    & \multicolumn{3}{c}{\textbf{Ours}} & \multicolumn{2}{c}{(rec)} & \multicolumn{2}{c}{(assert)} \\ \cmidrule(r){2-4} \cmidrule(r){5-6} \cmidrule(r){7-8}
    & \textbf{Spacer} & \quad\textbf{\hoice} & \quad\textbf{Eldarica} & \quad\textbf{Z3} & \quad\textbf{CVC4} & \quad\textbf{Z3} &\quad\textbf{CVC4} \\ \midrule
\quad prefix\_exists&0.147&0.141&0.628 & &  & & \\
\quad prefix\_exists\_invalid&0.192&0.202&1.052&\ 0.016&0.01 & & \\
\quad prefix\_forall&0.091&0.087&0.568 & &  & & \\
\quad prefix\_forall\_invalid&0.139&0.162&0.958& 0.019&0.011 & & \\
\quad prefix\_hd&0.214&0.212&0.697&0.018&0.009&\ 0.019&0.005\\
\quad prefix\_length & & 1.778&2.58&0.016&0.010&0.015&0.007\\
\quad prefix\_non\_nil&0.06&0.062&0.554&0.020&0.010&0.019&0.007\\
\quad prefix\_nth & & 1.732&5.454 & &  & & \\
\quad prefix\_sum&1.319&20.028&2.612 & &  & & \\
\quad prefix\_sum\_invalid&1.363&1.641&2.622&0.041&0.034 & & \\
\quad prefix\_trans&0.134&0.083&0.675 & &  & & \\
\quad scan\_sum\_length & & 2.043&3.058 & &  & & \\
\quad scan\_sum\_pos&0.214&0.199&1.232 & &  & & \\
\quad scan\_sum\_pos\_invalid&0.212&0.204&1.296&0.043&0.048 & & \\
\quad sorted&0.999&1.153&2.322 & &  & & \\
\quad sorted\_forall&0.175&0.211&0.988 & &  & & \\
\quad sorted\_neg&0.424&0.505&1.439 & &  & & \\
\quad sorted\_nth&0.565&0.721&1.706 & &  & & \\
\quad sorted\_pos&0.32&0.404&1.243 & &  & & \\
\quad sorted\_prefix&0.152&0.165&1.266 & &  & & \\
\quad sorted\_singleton&0.208&0.271&0.89&0.017&0.013&0.019&\\
\quad sumfold\_gte & &  & &  & & &\\
\quad take\_cons & & 3.306&3.908&0.039&0.03&0.017&0.010\\
\quad take\_exists&0.385&0.771&1.834 & &  & & \\
\quad take\_length & & 4.801&4.203 & &  & & \\
\quad take\_length\_invalid&1.289&1.664&3.048&0.019&0.008 & & \\
\quad take\_nil&0.169&0.118&0.707&0.039&0.028&0.018&0.008\\
\quad take\_sum&3.523&47.381&6.189 & &  & & \\
\quad take\_total\_count & & &7.01 & &  & & \\
\quad take\_total\_length\_eq & & &10.85 & &  & & \\
\quad take\_total\_length\_invalid&6.36&55.695&9.545&0.018&0.012 & & \\
\quad take\_total\_length\_leq & & 11.185&6.054 & &  & & \\
\textbf{CHC} \\
\quad compare\_length&1.174&1.328&5.629 & &  & & \\
\quad fib & & 16.824&20.671 & &  & & \\
\quad filter & & 4.763&10.414 & & &0.035&0.013\\
\quad insert\_sort&3.013&28.891&13.806 & & 0.048 & & \\
\quad insertion&3.786&8.14&12.265 & & &0.018&\\
\quad leq\_list&0.407&0.66&2.428 & &  & & \\
\quad make\_length & & 3.84&9.109&0.017&0.022&0.013&0.022\\
\quad merge\_sort&5.185 & & 33.017 & & 0.08 & & \\
\quad rev\_length & & 4.746&12.024 & & &0.023&0.011\\
\quad same&0.571&0.628&3.809 & &  & & \\
\quad sorted&0.9&1.473&4.967 & &  & & \\
\quad sum&0.809&1.129&4.797 & &  & & \\
\bottomrule
  \end{tabular}
}
\end{table}

%% file: main.bbl
\begin{thebibliography}{10}
\providecommand{\url}[1]{\texttt{#1}}
\providecommand{\urlprefix}{URL }
\providecommand{\doi}[1]{https://doi.org/#1}

\bibitem{BarrettCDHJKRT11}
Barrett, C.W., Conway, C.L., Deters, M., Hadarean, L., Jovanovic, D., King, T.,
  Reynolds, A., Tinelli, C.: {CVC4}. In: Gopalakrishnan, G., Qadeer, S. (eds.)
  Computer Aided Verification - 23rd International Conference, {CAV} 2011,
  Snowbird, UT, USA, July 14-20, 2011. Proceedings. Lecture Notes in Computer
  Science, vol.~6806, pp. 171--177. Springer (2011).
  \doi{10.1007/978-3-642-22110-1\_14}

\bibitem{DBLP:journals/jsat/BarrettST07}
Barrett, C.W., Shikanian, I., Tinelli, C.: An abstract decision procedure for a
  theory of inductive data types. J. Satisf. Boolean Model. Comput.
  \textbf{3}(1-2),  21--46 (2007). \doi{10.3233/sat190028}

\bibitem{Bjorner15}
Bj{\o}rner, N., Gurfinkel, A., McMillan, K.L., Rybalchenko, A.: Horn clause
  solvers for program verification. In: Fields of Logic and Computation {II} -
  Essays Dedicated to Yuri Gurevich on the Occasion of His 75th Birthday. LNCS,
  vol.~9300, pp. 24--51. Springer (2015). \doi{10.1007/978-3-319-23534-9\_2}

\bibitem{BlumensathGradel00}
Blumensath, A., Gr{\"{a}}del, E.: Automatic structures. In: 15th Annual {IEEE}
  Symposium on Logic in Computer Science, Santa Barbara, California, USA, June
  26-29, 2000. pp. 51--62. {IEEE} Computer Society (2000).
  \doi{10.1109/LICS.2000.855755}

\bibitem{DBLP:conf/vmcai/BradleyMS06}
Bradley, A.R., Manna, Z., Sipma, H.B.: What's decidable about arrays? In:
  Emerson, E.A., Namjoshi, K.S. (eds.) Verification, Model Checking, and
  Abstract Interpretation, 7th International Conference, {VMCAI} 2006,
  Charleston, SC, USA, January 8-10, 2006, Proceedings. Lecture Notes in
  Computer Science, vol.~3855, pp. 427--442. Springer (2006).
  \doi{10.1007/11609773\_28}

\bibitem{ChampionCKS20}
Champion, A., Chiba, T., Kobayashi, N., Sato, R.: {ICE}-based refinement type
  discovery for higher-order functional programs. J. Autom. Reason.
  \textbf{64}(7),  1393--1418 (2020). \doi{10.1007/s10817-020-09571-y}

\bibitem{DBLP:journals/fmsd/DAntoniV15}
D'Antoni, L., Veanes, M.: Extended symbolic finite automata and transducers.
  Formal Methods Syst. Des.  \textbf{47}(1),  93--119 (2015).
  \doi{10.1007/s10703-015-0233-4}

\bibitem{DBLP:conf/cav/DAntoniV17}
D'Antoni, L., Veanes, M.: The power of symbolic automata and transducers. In:
  Majumdar, R., Kuncak, V. (eds.) Computer Aided Verification - 29th
  International Conference, {CAV} 2017, Heidelberg, Germany, July 24-28, 2017,
  Proceedings, Part {I}. Lecture Notes in Computer Science, vol. 10426, pp.
  47--67. Springer (2017). \doi{10.1007/978-3-319-63387-9\_3}

\bibitem{DBLP:journals/tplp/AngelisFPP18a}
{De Angelis}, E., Fioravanti, F., Pettorossi, A., Proietti, M.: Solving horn
  clauses on inductive data types without induction. {TPLP}  \textbf{18}(3-4),
  452--469 (2018). \doi{10.1017/S1471068418000157}

\bibitem{DeAngelisFPP20}
{De Angelis}, E., Fioravanti, F., Pettorossi, A., Proietti, M.: Removing
  algebraic data types from constrained horn clauses using difference
  predicates. In: Peltier, N., Sofronie{-}Stokkermans, V. (eds.) Automated
  Reasoning - 10th International Joint Conference, {IJCAR} 2020, Paris, France,
  July 1-4, 2020, Proceedings, Part {I}. Lecture Notes in Computer Science,
  vol. 12166, pp. 83--102. Springer (2020). \doi{10.1007/978-3-030-51074-9\_6}

\bibitem{DixonFleuriot03}
Dixon, L., Fleuriot, J.D.: {IsaPlanner}: A prototype proof planner in
  {Isabelle}. In: Baader, F. (ed.) Automated Deduction - CADE-19, 19th
  International Conference on Automated Deduction Miami Beach, FL, USA, July 28
  - August 2, 2003, Proceedings. Lecture Notes in Computer Science, vol.~2741,
  pp. 279--283. Springer (2003). \doi{10.1007/978-3-540-45085-6\_22}

\bibitem{DBLP:journals/pacmpl/EzudheenND0M18}
Ezudheen, P., Neider, D., D'Souza, D., Garg, P., Madhusudan, P.: Horn-{ICE}
  learning for synthesizing invariants and contracts. Proc. {ACM} Program.
  Lang.  \textbf{2}({OOPSLA}),  131:1--131:25 (2018). \doi{10.1145/3276501}

\bibitem{DBLP:conf/tacas/FedyukovichE21}
Fedyukovich, G., Ernst, G.: Bridging arrays and adts in recursive proofs. In:
  Groote, J.F., Larsen, K.G. (eds.) Tools and Algorithms for the Construction
  and Analysis of Systems - 27th International Conference, {TACAS} 2021, Held
  as Part of the European Joint Conferences on Theory and Practice of Software,
  {ETAPS} 2021, Luxembourg City, Luxembourg, March 27 - April 1, 2021,
  Proceedings, Part {II}. Lecture Notes in Computer Science, vol. 12652, pp.
  24--42. Springer (2021). \doi{10.1007/978-3-030-72013-1\_2}

\bibitem{DBLP:conf/lics/Gradel20}
Gr{\"{a}}del, E.: Automatic structures: Twenty years later. In: Hermanns, H.,
  Zhang, L., Kobayashi, N., Miller, D. (eds.) {LICS} '20: 35th Annual
  {ACM/IEEE} Symposium on Logic in Computer Science, Saarbr{\"{u}}cken,
  Germany, July 8-11, 2020. pp. 21--34. {ACM} (2020).
  \doi{10.1145/3373718.3394734}

\bibitem{HashimotoUnno15}
Hashimoto, K., Unno, H.: Refinement type inference via horn constraint
  optimization. In: Blazy, S., Jensen, T.P. (eds.) Static Analysis - 22nd
  International Symposium, {SAS} 2015, Saint-Malo, France, September 9-11,
  2015, Proceedings. Lecture Notes in Computer Science, vol.~9291, pp.
  199--216. Springer (2015). \doi{10.1007/978-3-662-48288-9\_12}

\bibitem{HaudebourgThesis}
Haudebourg, T.: Automatic verification of higher-order functional programs
  using regular tree languages. Ph.D. thesis, Universit\'{e}x Rennes (2020)

\bibitem{HojjatRummer18}
Hojjat, H., R{\"{u}}mmer, P.: The {ELDARICA} horn solver. In: Bj{\o}rner, N.,
  Gurfinkel, A. (eds.) 2018 Formal Methods in Computer Aided Design, {FMCAD}
  2018, Austin, TX, USA, October 30 - November 2, 2018. pp.~1--7. {IEEE}
  (2018). \doi{10.23919/FMCAD.2018.8603013}

\bibitem{JohanssonDB10}
Johansson, M., Dixon, L., Bundy, A.: Case-analysis for rippling and inductive
  proof. In: Kaufmann, M., Paulson, L.C. (eds.) Interactive Theorem Proving,
  First International Conference, {ITP} 2010, Edinburgh, UK, July 11-14, 2010.
  Proceedings. Lecture Notes in Computer Science, vol.~6172, pp. 291--306.
  Springer (2010)

\bibitem{DBLP:conf/lcc/KhoussainovN94}
Khoussainov, B., Nerode, A.: Automatic presentations of structures. In:
  Leivant, D. (ed.) Logical and Computational Complexity. Selected Papers.
  Logic and Computational Complexity, International Workshop {LCC} '94,
  Indianapolis, Indiana, USA, 13-16 October 1994. Lecture Notes in Computer
  Science, vol.~960, pp. 367--392. Springer (1994).
  \doi{10.1007/3-540-60178-3\_93}

\bibitem{KomuravelliGC16}
Komuravelli, A., Gurfinkel, A., Chaki, S.: {SMT}-based model checking for
  recursive programs. Formal Methods Syst. Des.  \textbf{48}(3),  175--205
  (2016). \doi{10.1007/s10703-016-0249-4}

\bibitem{Kowalski74}
Kowalski, R.A.: Predicate logic as programming language. In: Rosenfeld, J.L.
  (ed.) Information Processing, Proceedings of the 6th {IFIP} Congress 1974,
  Stockholm, Sweden, August 5-10, 1974. pp. 569--574. North-Holland (1974)

\bibitem{Minsky61}
Minsky, M.L.: Recursive unsolvability of post's problem of ``tag'' and other
  topics in theory of turing machines. Annals of Mathematics pp. 437--455
  (1961)

\bibitem{MouraBjorner08}
de~Moura, L.M., Bj{\o}rner, N.: {Z3:} an efficient {SMT} solver. In:
  Ramakrishnan, C.R., Rehof, J. (eds.) Tools and Algorithms for the
  Construction and Analysis of Systems, 14th International Conference, {TACAS}
  2008, Held as Part of the Joint European Conferences on Theory and Practice
  of Software, {ETAPS} 2008, Budapest, Hungary, March 29-April 6, 2008.
  Proceedings. Lecture Notes in Computer Science, vol.~4963, pp. 337--340.
  Springer (2008). \doi{10.1007/978-3-540-78800-3\_24}

\bibitem{Paulson94}
Paulson, L.C.: {Isabelle} - A Generic Theorem Prover (with a contribution by
  {T. Nipkow}), Lecture Notes in Computer Science, vol.~828. Springer (1994).
  \doi{10.1007/BFb0030541}

\bibitem{ShimodaKSS21}
Shimoda, T., Kobayashi, N., Sakayori, K., Sato, R.: Symbolic automatic
  relations and their applications to {SMT} and {CHC} solving [data set]
  (2021). \doi{10.5281/zenodo.5140576}

\bibitem{UnnoKobayashi09}
Unno, H., Kobayashi, N.: Dependent type inference with interpolants. In: Porto,
  A., L{\'{o}}pez{-}Fraguas, F.J. (eds.) Proceedings of the 11th International
  {ACM} {SIGPLAN} Conference on Principles and Practice of Declarative
  Programming, September 7-9, 2009, Coimbra, Portugal. pp. 277--288. {ACM}
  (2009). \doi{10.1145/1599410.1599445}

\bibitem{UnnoTS17}
Unno, H., Torii, S., Sakamoto, H.: Automating induction for solving horn
  clauses. In: Majumdar, R., Kuncak, V. (eds.) Computer Aided Verification -
  29th International Conference, {CAV} 2017, Heidelberg, Germany, July 24-28,
  2017, Proceedings, Part {II}. Lecture Notes in Computer Science, vol. 10427,
  pp. 571--591. Springer (2017). \doi{10.1007/978-3-319-63390-9\_30}

\bibitem{DBLP:conf/lpar/VeanesBM10}
Veanes, M., Bj{\o}rner, N., de~Moura, L.M.: Symbolic automata constraint
  solving. In: Ferm{\"{u}}ller, C.G., Voronkov, A. (eds.) Logic for
  Programming, Artificial Intelligence, and Reasoning - 17th International
  Conference, LPAR-17, Yogyakarta, Indonesia, October 10-15, 2010. Proceedings.
  Lecture Notes in Computer Science, vol.~6397, pp. 640--654. Springer (2010).
  \doi{10.1007/978-3-642-16242-8\_45}

\bibitem{VeanesHT10}
Veanes, M., de~Halleux, P., Tillmann, N.: Rex: Symbolic regular expression
  explorer. In: Third International Conference on Software Testing,
  Verification and Validation, {ICST} 2010, Paris, France, April 7-9, 2010. pp.
  498--507. {IEEE} Computer Society (2010). \doi{10.1109/ICST.2010.15}

\bibitem{ZhuJagannathan13}
Zhu, H., Jagannathan, S.: Compositional and lightweight dependent type
  inference for {ML}. In: Giacobazzi, R., Berdine, J., Mastroeni, I. (eds.)
  Verification, Model Checking, and Abstract Interpretation, 14th International
  Conference, {VMCAI} 2013, Rome, Italy, January 20-22, 2013. Proceedings.
  Lecture Notes in Computer Science, vol.~7737, pp. 295--314. Springer (2013).
  \doi{10.1007/978-3-642-35873-9\_19}

\end{thebibliography}
